\documentclass[10pt,aps,pra,twocolumn,floatfix,nofootinbib,superscriptaddress,longbibliography]{revtex4-2}
\usepackage[utf8]{inputenc}

\usepackage[dvipsnames]{xcolor}
\usepackage{bm,graphicx,mathrsfs,amsmath,amssymb,mathtools,makecell,bbm,amsthm,dsfont,color,times,txfonts,nicefrac,framed,enumitem,tikz,physics,wrapfig,amsfonts,lipsum,nicematrix,xargs, multirow}
\definecolor{equationcolor}{RGB}{222,94,100}
\definecolor{changescolor}{rgb}{0, 0, 0.7}
\newcommand{\iden}{\mathbbm{1}}
\usepackage[colorlinks=true,linkcolor=changescolor,citecolor=equationcolor]{hyperref}
\hypersetup{urlcolor=equationcolor}
\usetikzlibrary{arrows.meta}
\usetikzlibrary{positioning}
\usetikzlibrary{calc} 
\usepackage{comment}
\graphicspath{ {Figures} }
\usepackage{transparent} 
\usepackage{import}
\newtheorem{definition}{Definition}[section]
\newtheorem{lemma}[definition]{Lemma}

\newtheorem{corollary}[definition]{Corollary}
\newtheorem{proposition}[definition]{Proposition}

\newcommand{\qh}{\hat q}
\newcommand{\ph}{\hat p}
\newcommand{\rhoh}{\hat\rho}
\newcommand{\Pih}{\hat\Pi}
\newcommand{\Dh}{\hat D}
\newcommand{\Wop}[2]{\hat W(#1,#2)}
\newcommand{\Sop}[2]{\hat S_{#1,#2}}
\newcommand{\ZL}{\hat Z_L}
\newcommand{\XL}{\hat X_L}
\newcommand{\YL}{\hat Y_L}
\newcommand{\GKP}{\mathrm{GKP}}
\renewcommand{\d}{\textrm{d}}

\begin{document}

\title{The logical set of Gaussian states under the stabilizer subsystem decomposition}
\author{Santiago Zamora}
\affiliation{International Institute of Physics, Federal University of Rio Grande do Norte, 59078-970, Natal, Brazil}
\affiliation{Physics Department, Federal University of Rio Grande do Norte, Natal, 59072-970, Rio Grande do Norte, Brazil}
\affiliation{Department of Mathematics and Statistics, University of Ottawa, Ottawa, Ontario K1N 6N5, Canada}
\author{Kaustav Chatterjee}
\author{Ulrik Lund Andersen}
\affiliation{Center for Macroscopic Quantum States bigQ, Department of Physics,
Technical University of Denmark, Fysikvej 307, 2800 Kgs. Lyngby, Denmark}
\author{Thiago Lucena de Macedo Guedes}
\affiliation{International Institute of Physics, Federal University of Rio Grande do Norte, 59078-970, Natal, Brazil}
\author{Jonatan Bohr Brask}
\affiliation{Center for Macroscopic Quantum States bigQ, Department of Physics,
Technical University of Denmark, Fysikvej 307, 2800 Kgs. Lyngby, Denmark}
\author{A.~de~Oliveira~Junior}
\affiliation{Center for Macroscopic Quantum States bigQ, Department of Physics,
Technical University of Denmark, Fysikvej 307, 2800 Kgs. Lyngby, Denmark}
\author{Rafael Chaves }
\affiliation{International Institute of Physics, Federal University of Rio Grande do Norte, 59078-970, Natal, Brazil}

\begin{abstract}
Universal quantum computation requires non-Gaussianity in continuous-variable systems and non-stabilizerness in discrete-variable systems.
Yet, mapping a Gaussian state into the Gottesman--Kitaev--Preskill logical subspace can yield a logical qubit with non-stabilizer resources. To understand how the continuous-variable structure determines logical resources, we characterize the image of single-mode Gaussian states under the corresponding extraction channel, known as the stabilizer subsystem decomposition. Specifically, we derive series expressions for the logical Bloch vector and characterize the resulting reachable set, enabling an analytical treatment of the robustness of magic of the logical states. We then turn this geometric framework into a certification tool: by mapping a discrete-variable qubit witness to a continuous-variable squeezing witness, we obtain lower bounds on the squeezing of prepared states. We also show that ensembles of Gaussian states can violate both stabilizer and classical bounds. Violation of the stabilizer bound certifies the relational resource known as ``set-magic,'' while violation of the classical bound enables the certification of cryptographic randomness at moderate squeezing.
\end{abstract}

\maketitle
\textit{Introduction---}Continuous-variable (CV) and discrete-variable (DV) systems rely on different resources for computation. While non-Gaussianity is required for universal quantum computation in CV systems~\cite{MariEisert2012, Veitch2013,Takagi2018,Walschaers2021}, non-stabilizerness, or ``magic,'' plays an analogous role in DV systems~\cite{BravyiKitaev2005, Veitch2014}. The question is then how these two resources are related. The Gottesman-Kitaev-Preskill (GKP) code~\cite{GKP2001} provides a bridge between these two paradigms by encoding discrete logical information into a continuous-variable bosonic mode~\cite{Hahn2022,Hahn2025}. Less understood is the reverse direction. First, what logical qubit states can Gaussian bosonic states produce when mapped into the logical subspace? Second, what restrictions does Gaussianity impose on their logical non-stabilizerness? 

To address these questions, we consider ideal GKP error correction followed by extraction of the logical qubit, with the syndrome discarded. The resulting channel is described by the stabilizer subsystem decomposition (SSD)~\cite{Shaw_2024,Calcluth_2024}. Previous numerical studies showed that Gaussian inputs can produce logical non-stabilizer states under this map, remarkably even when the input is the vacuum~\cite{Calcluth_2024}. What remains unclear, however, is the structure of the full family of logical states accessible from Gaussian inputs and, consequently, the constraints that Gaussianity imposes on the resulting logical resource. Related work on non-Gaussian states further emphasizes the broader question of how properties of a bosonic state are inherited by the logical state obtained after mapping to the GKP subspace.~\cite{David2026}.

In this Letter, we derive an exact parametrization of the logical image of arbitrary single-mode Gaussian states under SSD. The resulting series are absolutely convergent for every finite positive-definite covariance matrix. We then determine the $XZ$ projection of this image analytically, prove that it is convex, and identify the Gaussian states that generate its boundary. As a consequence, we obtain the exact robustness of the mapped vacuum and a leading-harmonic stationary estimate within an undisplaced diagonal pure-state family, consistent with earlier numerical work~\cite{Calcluth_2024}. We also construct fixed-squeezing Gaussian families approaching all six vertices of the single-qubit stabilizer polytope. Their convex combinations approximate every point of the polytope with a one-sided error of order $e^{-r}$, where $r$ is the squeezing parameter. 

The exact component bounds also provide a tight squeezing criterion for a single mapped Gaussian preparation, without optimizing over displacements or squeezing angles. To demonstrate the operational utility of our framework, we explore a prepare-and-measure (PAM) scenario~\cite{brask2026quantum}. Using a qubit witness~\cite{Taoutioui2025}, we map the corresponding DV inequality to a CV squeezing witness. We also show that logical ensembles derived from Gaussian states can violate the stabilizer bound of this witness. Such violations certify ``set-magic''~\cite{salazar2022resourcetheoryabsolutenegativity,Wagner2025}, a relational resource with direct applications in enhancing communication protocols~\cite{Zamora2025}. Finally, violations of the classical PAM bound enable the certification of cryptographic randomness from Gaussian inputs, giving a nonzero certified conditional min-entropy already at moderate levels of squeezing.

\textit{Framework---}We consider a single bosonic mode with quadratures $\qh,\ph$ satisfying $[\qh,\ph]=i$. A single-mode Gaussian state is fully specified by its displacement $\boldsymbol d=(q_0,p_0)^T$, with $q_0=\langle\qh\rangle$ and $p_0=\langle\ph\rangle$, and its covariance matrix $V$~\cite{serafini2017}. In $(q,p)$ order, the covariance entries are $V_{qq}\!=\!\langle\qh^2\rangle-q_0^2$, $V_{pp}\!=\!\langle\ph^2\rangle-p_0^2$, and $V_{qp}\!=\!V_{pq}\!=\!\frac12\langle\qh\ph+\ph\qh\rangle-q_0p_0$, where $\langle\hat O\rangle=\Tr(\rhoh\hat O)$. The covariance matrix is real, symmetric, positive definite, and satisfies the uncertainty relation $\det V\geq\tfrac14$. For a general single-mode Gaussian state, $V$ can be parametrized as \mbox{$V=\nu R(\theta)S(2r)R(\theta)^T$}, where $\nu=\sqrt{\det V}\geq\tfrac12$ is the symplectic eigenvalue of $V$, $S(r) = \operatorname{diag}(e^{-r}, e^{r})$ is the squeezing matrix with $r\geq0$ and $R(\theta) = e^{-i\theta \sigma_y}$ is the rotation matrix with $\theta\in[0,\pi)$. Pure states correspond to $\nu=\tfrac12$~\cite{Weedbrook_2012}. 

To extract a logical qubit, we apply an ideal GKP syndrome measurement and correction and discard the syndrome. For the square-lattice GKP code, the resulting channel is associated with the stabilizer subsystem decomposition (SSD)~\cite{Shaw_2024,Calcluth_2024} and can be written in the usual ideal-code notation as
\begin{equation}
\Lambda(\hat\rho)=  \frac1{\sqrt\pi}\int_{-\sqrt\pi/2}^{\sqrt\pi/2} \d s_q\int_{-\sqrt\pi/2}^{\sqrt\pi/2} \d s_p\ \Pih\Dh(-\vec{s})\hat\rho\Dh(-\vec{s})^\dagger\Pih, \label{eq:SSDmap_mt}
\end{equation}
where $\vec s = (s_q,s_p)$, $\hat{D}(\vec{s}) = e^{i \hat{q} s_p} e^{-i \hat{p} s_q}$ translates $\qh$ by $s_q$ and $\ph$ by $s_p$, and $\hat{\Pi} =\ket{0_\GKP}\bra{0_\GKP}+\ket{1_\GKP}\bra{1_\GKP}$ is the projector onto the code space spanned by the ideal GKP states $|l_\GKP\rangle = \sum_{n\in\mathbb Z} \big|\qh=(2n+l)\sqrt\pi\big\rangle$, for $l=0,1$. Equation~\eqref{eq:SSDmap_mt} is understood through its two-dimensional logical matrix elements, as detailed in the Supplemental Material (SM)~\ref{sec:SSD}, which also establishes the equivalence of ordered and symmetric displacement conventions. The displacements $\Dh(\sqrt\pi,0)$ and $\Dh(0,\sqrt\pi)$ induce the logical Pauli operators $\XL$ and $\ZL$, respectively, with $\YL=i\XL\ZL$. The logical state is
\begin{equation}
\hat\rho_L=\Lambda(\hat\rho) =\frac12\left(\iden_L+X\hat X_L+Y \hat Y_L+Z  \hat Z_L\right),
\end{equation}
where $X=\operatorname{Tr}(\rhoh_L\XL)$ and analogously for $Y$ and $Z$. We denote its Bloch vector by $\vec R_\rho=(X,Y,Z)$.

Our goal is to characterize the logical states reachable from Gaussian inputs under fixed or bounded squeezing. We define
\begin{equation}
    G_{\leq r_0}= \left\{(X,Y,Z):\hat\rho\ {\rm is\ Gaussian\ with\ squeezing}\ r\leq r_0\right\},
\end{equation}
and write $G_{r_0}$ for the corresponding set at fixed squeezing $r=r_0$. With this notation $G_{\leq\infty}$ represents the image of all Gaussian states.

\textit{Gaussian states under SSD---}The central geometric question is which logical Bloch vectors are compatible with a Gaussian input. To express the variance and displacement dependence of the logical components, we introduce
\begin{equation}\label{eq:A_mt}
    A(t,\varphi)=\frac4\pi\sum_{m=0}^{\infty}\frac{(-1)^m}{2m+1}e^{-\frac{\pi}{2}(2m+1)^2t}\cos[(2m+1)\varphi],
\end{equation}
for $t>0$ and $\varphi\in\mathbb R$, and write $A(t)\equiv A(t,0)$ for its zero-displacement value. The exact logical Bloch components are
\begin{align}
    X & = A(V_{pp},\sqrt\pi p_0), \nonumber \\
    Y & =-\frac8{\pi^2}\sum_{m=0}^{\infty}\sum_{n=-\infty}^{\infty}\frac{\cos\!\big(\sqrt\pi[(2m+1)q_0+(2n+1)p_0]\big)}{(2m+1)(2n+1)}\varphi_{mn}, \nonumber \\
    Z & = A(V_{qq},\sqrt\pi q_0),\label{Eq:XYZ-component}
\end{align}
where $\varphi_{mn} = e^{-\frac{\pi}{2}[(2m+1)^2V_{qq}+(2n+1)^2V_{pp} +2(2m+1)(2n+1) V_{pq}]}$ is a Gaussian damping factor. These series converge absolutely for every finite positive-definite covariance matrix (their derivation is given in SM~\ref{sec:Bloch_vec_ssd} and~\ref{sec:Gaussian_Set}). The projection of this image onto the $XZ$ plane can be characterized exactly, as we now show.

\begin{figure}[t!]
    \centering
    \includegraphics{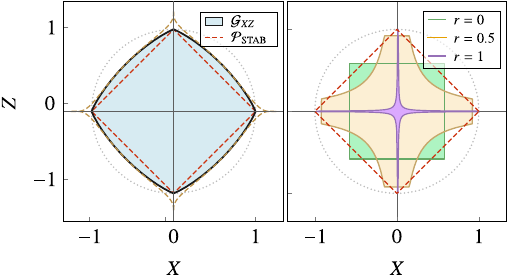}
    \caption{\textbf{Gaussian images projected onto the $XZ$ plane.} \textit{Left:} Projection $\mathcal G_{XZ}$ for all Gaussian inputs (light blue), with its exact boundary (black). The orange dashed curve is the leading-harmonic approximation in Eq.~\eqref{eq:approx_XZ_mt}. \textit{Right:} Exact projections for fixed squeezing $r=0,0.5,1$. In both panels, the gray dotted circle is the boundary of the Bloch disk and the red dashed diamond is $\operatorname{proj}_{XZ}\mathcal P_{\rm STAB}$.}
    \label{fig:XZ_set}
\end{figure}

We denote the projection onto the $XZ$ plane by $\mathcal G_{XZ}\equiv \operatorname{proj}_{XZ}(G_{\leq\infty})$. For a fixed covariance matrix $V$, varying the displacements $p_0$ and $q_0$ allows $X$ and $Z$ to range independently over $[-A(V_{pp}),A(V_{pp})]$ and $[-A(V_{qq}),A(V_{qq})]$, respectively. The reachable pairs $(X,Z)$ therefore form a rectangle. Taking the union of these rectangles over all positive variance pairs satisfying $V_{pp}V_{qq}\geq\tfrac14$ gives $\mathcal G_{XZ}$, and the resulting set is convex. For $X,Z>0$, its boundary is the strictly concave curve $(X,Z)=\bigl(A(V_{pp}),A(\tfrac{1}{4V_{pp}})\bigr)$, where $V_{pp}>0$. The curve is generated by pure, axis-aligned Gaussian inputs whose displacements lie on the lattice $2\sqrt\pi\mathbb Z$. In particular, centered Gaussian states attain this boundary.  The remaining quadrants follow by reflection, and the four axis endpoints are attained only as limits. The proof and boundary-state characterization are given in SM~\ref{sec:exact_XZ}.

For fixed squeezing $r$, mixed states do not enlarge the projection. At fixed $r$ and $\theta$, increasing $\nu$ increases both $V_{qq}$ and $V_{pp}$ and hence decreases the allowed ranges of $X$ and $Z$. Consequently, it is sufficient to consider pure states, for which $V_{qq}=\cosh(2r)-V_{pp}$ and $\tfrac12e^{-2r}\leq V_{pp}\leq\tfrac12e^{2r}$. Varying $\theta$ gives all variance pairs in this interval. The projection is then obtained from the same rectangle construction. In particular, for $r=0$ it is the square $[-A(\tfrac12),A(\tfrac12)]^2$.

Let $\mathcal P_{\rm STAB}=\operatorname{conv}\{(\pm1,0,0),(0,\pm1,0),(0,0,\pm1)\}$ be the single-qubit stabilizer polytope. Here $\operatorname{conv}\{S\}$ denotes the convex hull of the set $S$.  Figure~\ref{fig:XZ_set} shows the two $XZ$ projections discussed above. In the left panel, part of $\mathcal G_{XZ}$ lies outside the stabilizer diamond, showing that Gaussian inputs can produce non-stabilizer logical states; the leading-harmonic approximation shown in Eq.~\eqref{eq:approx_XZ_mt} follows the exact boundary closely except near the coordinate axes, where it leaves the Bloch disk. The right panel shows that the fixed-$r$ projections become concentrated near the coordinate axes as $r$ increases. The exact projection also gives an immediate consequence for logical magic. For a single qubit, the robustness of magic is $R^{(1)}(\rho_L)=\max\{1,|X|+|Y|+|Z|\}$~\cite{Calcluth_2024} and quantifies the overhead of quasiprobability-sampling protocols~\cite{howard2017application}. For the undisplaced vacuum, $V_{qq}=V_{pp}=\tfrac12$ and $V_{qp}=0$, so $X=Z=A(\tfrac12)$ and $Y=0$. Its exact robustness is therefore $R^{(1)}(\rho_L^{\rm vac})=2A(\tfrac12)\approx1.160314$, consistent with the numerical value $1.160$ in Ref.~\cite{Calcluth_2024}. More generally, we prove in the SM~\ref{sec:proof_thm1} that $|X|+|Z|\leq2A(\tfrac12)$ for every Gaussian input, with equality only for the vacuum covariance and its lattice-displaced copies. Thus, squeezing cannot increase $|X|+|Z|$ beyond the vacuum value. Any larger robustness must come from a nonzero $Y$ component. Since the exact $XZ$ projection does not determine $Y$, the optimization over all Gaussian inputs remains a three-dimensional problem. Using the leading-harmonic approximation for the three components (see Eq.~\ref{eq:Bv_r_th_mt}), the calculation in SM~\ref{sec:ROM_gaussian_States} gives $r\simeq0.262$ and $R^{(1)}\simeq1.302$. These approximate values agree with the numerical maximum $1.303$ reported in Ref.~\cite{Calcluth_2024}.

The exact series reveals that the logical images of Gaussian states asymptotically approach the vertices of the stabilizer polytope as squeezing increases. More precisely, for each vertex of $\mathcal P_{\rm STAB}$, there is a family of Gaussian inputs whose Bloch vectors under SSD approach that vertex as $r\to\infty$. Consequently, every point of $\mathcal P_{\rm STAB}$ can be approximated by a convex combination of Bloch vectors in $G_r$, with an error bounded asymptotically by
\begin{equation}
    \sup_{\vec w\in\mathcal P_{\rm STAB}} d\!\left(\vec w,\operatorname{conv}(G_r)\right) \leq\frac{2\sqrt2}{\pi}e^{-r}\bigl(1+o(1)\bigr),
    \label{eq:thm_rate}
\end{equation}
where $d(\vec w,S)=\inf_{\vec v\in S}\|\vec w-\vec v\|_2$ is the Euclidean distance from $\vec w$ to the set $S$, and $o(1)\to0$ as $r\to\infty$. The converse containment is treated in SM~\ref{sec:proof_thm1}: taking $r\to\infty$ at fixed squeezing angle, displacement, and symplectic eigenvalue, every limiting Bloch vector lies on a coordinate axis. Furthermore, explicit pure Gaussian families approach all six stabilizer vertices: axis-aligned squeezing gives a doubly exponential approach to the $X$ and $Z$ vertices, while diagonal squeezing gives a single exponential approach to the $Y$ vertices. Convex combinations of these vertex approximations give Eq.~\eqref{eq:thm_rate} (SM~\ref{sm:thm_rate}).  Thus, $\operatorname{conv}(G_\infty)=\mathcal P_{\rm STAB}$, where $G_\infty =\lim_{r\to\infty}G_r$.

Unlike $X$ and $Z$, the component $Y$ depends on the off-diagonal covariance $V_{pq}$ and therefore couples the two quadratures. To obtain a simple approximation to the three-dimensional image, we retain  the $m=0$ term of $X$ and $Z$ and the $(m,n)=(0,0)$ and $(0,-1)$ terms of $Y$ in Eq.~\eqref{Eq:XYZ-component}. For undisplaced inputs, $q_0=p_0=0$, this gives
\begin{equation}
\vec R_\rho\simeq\frac4\pi\qty(e^{-\frac{\pi V_{pp}}2}, \frac4\pi e^{-\frac\pi2(V_{qq}+V_{pp})}\sinh(\pi V_{pq}), e^{-\frac{\pi V_{qq}}2}).
\label{eq:Bv_r_th_mt}
\end{equation}
Within this approximation, fixing $X$ and $Z$ fixes $V_{pp}$ and $V_{qq}$. Since $|Y|$ increases with $|V_{pq}|$, its largest value is obtained for a pure state, for which
$V_{pq}^2=V_{qq}V_{pp}-\tfrac14$. Eliminating the covariance entries gives $Y\simeq\pm XZ\sinh\left(\sqrt{4\ln\Big(\frac{\pi X}{4}\Big)\ln\Big(\frac{\pi Z}{4}\Big)-\frac{\pi^2}{4}}\right)$. This equation gives the largest $|Y|$ within the undisplaced leading-harmonic approximation. The same approximation gives simple curves in the $XZ$ plane. For the pure, axis-aligned states that generate the exact boundary, $V_{pq}=0$ and $V_{qq}V_{pp}=\tfrac14$, giving
\begin{equation}\label{eq:approx_XZ_mt}
    \ln\Big(\frac{\pi X}{4}\Big) \ln\Big(\frac{\pi Z}{4}\Big) \simeq\frac{\pi^2}{16}.
\end{equation}
This corresponds to the orange curve in the left panel of Fig.~\ref{fig:XZ_set}. For pure states with fixed squeezing $r$, $V_{qq}+V_{pp}=\cosh(2r)$, and the approximation instead gives $XZ\simeq\frac{16}{\pi^2} e^{-\frac{\pi}{2}\cosh(2r)}$. Figure~\ref{fig:XZ_set} compares the exact $XZ$ boundary with its leading-harmonic approximation and shows the exact projections at fixed squeezing. The approximation closely follows the exact boundary away from the coordinate axes, while near the axes it can extend outside the Bloch disk. The exact projection extends beyond the stabilizer diamond, while the fixed-$r$ projections become increasingly concentrated near the axes as $r$ increases.

\textit{Certification from Gaussian geometry---}The preceding geometry can be read in reverse: observable logic-state statistics constrain the squeezing required of a Gaussian input. This gives an exact test for a single preparation and bounded-squeezing benchmarks for ensembles in a prepare-and-measure (PAM) scenario~\cite{brask2026quantum}.

For any Gaussian state with $0\le r\le r_0$, the fixed-variance bounds give (SM~\ref{sec:exact_XZ}) $\max(|X|,|Z|)\le A\left(\tfrac12e^{-2r_0}\right)$. The bound is attained by an undisplaced squeezed vacuum aligned with one of the two quadratures. Under the assumptions of Gaussian input states and faithful SSD extraction, a violation by either logical component certifies $r>r_0$, requiring no optimization over displacements or squeezing angles. The same idea can be extended to a PAM witness, combining measurement statistics from an ensemble of prepared states. Maximizing the witness over Gaussian preparations with $r\leq r_0$ gives a bound whose violation certifies that at least one preparation has squeezing greater than $r_0$. We consider a PAM scenario with three qubit preparations $\{\vec{r}_x\}_{x=0}^2$ and two dichotomic measurements $\{M_{b|y}\}_{y=0}^1$, and $b=0,1$. The following witness~\cite{Taoutioui2025} distinguishes classical-bit, stabilizer, and general qubit preparations:
\begin{equation}
S_3(t) = 2t\|\vec{r}_0+\vec{r}_1-\vec{r}_2\|_2 + 2(1-t)\|\vec{r}_0-\vec{r}_1\|_2,
\end{equation}
where $\vec{r}_0, \vec{r}_1, \vec{r}_2$ are the Bloch vectors of the prepared qubit states and $t \in [0,1]$ is a free parameter. It distinguishes them by having a different maximum value for each type of preparation. In particular, the classical and quantum bounds are given in Ref.~\cite{Taoutioui2025} and the stabilizer bound in Ref.~\cite{Zamora2025}. They read
\begin{align}
    S_3^C(t)&=\max\{4-2t,6t\}, \\ 
    S_3^Q(t)&=2t+4\sqrt{t^2+(1-t)^2},\\
    S_3^{\text{STAB}}(t)&=\max\{6t, 4-2t, 2t\sqrt{5}+2(1-t)\sqrt{2}\}. 
\end{align}
\begin{figure*}
    \centering
    \includegraphics{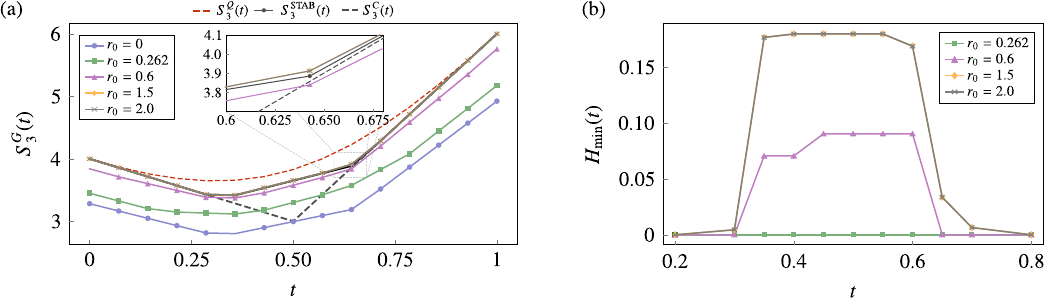}
    \caption{\textbf{Certification from Gaussian geometry.} (a) The red dashed, black solid, and black dashed lines show the qubit, stabilizer, and classical bounds, respectively. Colored markers show the numerically optimized values for Gaussian states with squeezing $r\leq r_0$ under the SSD map. The vacuum cannot violate the classical bound, so nonzero squeezing is required. The inset shows a violation of the stabilizer bound near $t\approx0.643$. (b) Worst-case conditional min-entropy $H_{\min}$ for the generation setting $(x=2,y=1)$ versus the tilting parameter $t$, for different squeezing bounds $r_0$. The result was obtained with the level $L(1,2)$ of the SDP relaxation introduced in Ref.~\cite{sarubi2026}. Since $S_3(t)$ does not fix all correlators, $H_{\min}$ is minimized over Gaussian configurations attaining the maximal $S_3(t)$. No randomness is certified for $r_0\leq0.262$. At higher squeezing it becomes nonzero and saturates at approximately $0.18$ bits for $r_0\geq1.5$.}
    \label{F-certification}
\end{figure*}

Our framework allows the PAM witness $S_3(t)$ to be mapped into a squeezing witness. Formally, for a maximum squeezing parameter $r_0$, we define the Gaussian value $S_3^G(t;r_0) \equiv \max_{\vec{r}_x \in G_{\le r_0}} S_3(t)$. Because the set of accessible states satisfies $G_{\le r_0} \subseteq G_{\le r_0'}$ for $r_0 < r_0'$, the bound $S_3^G(t;r_0)$ is a non-decreasing function of $r_0$. Consequently, for any $t$, an observed value exceeding this bound certifies a squeezing larger than $r_0$:
\begin{equation}
\!\!S_3(t)\!>\!S_3^G(t;r_0) \Longrightarrow \text{at least one preparation has }r>r_0. \label{eq:certification}
\end{equation}
Figure~\hyperref[F-certification]{\ref{F-certification}(a)} plots the $S_3(t)$ curves for $r \le r_0$. A wide gap between the bounds allows the PAM witness to discriminate weakly squeezed states; however, this gap closes at higher $r$, reducing sensitivity in the strongly squeezed regime.

As shown in the inset of Fig.~\hyperref[F-certification]{\ref{F-certification}(a)}, GKP-encoded Gaussian states can slightly violate the stabilizer bound $S_3^{\text{STAB}}$, allowing the certification of ``set-magic"~\cite{Wagner2025} within the prepared Gaussian states---a relational resource characterizing whether an ensemble of states can be rigidly rotated into the stabilizer polytope. As detailed in SM~\ref{sec: set magic}, for any $0<t<1$ there exists a finite squeezing value $r(t)=\log[S_3^{\mathrm{STAB}}(t)/k(t)]$ where $k(t)=4\min\{t/\sqrt{5},1-t\}$ and a pure preparation with this squeezing (or above) that certifies set-magic. The proof gives a consistent construction of such pure preparations whose logical states lie in the XZ plane, all having robustness of magic strictly less than the vacuum. These logical states correspond to the pure Gaussian states with covariances $\frac12\operatorname{diag }(e^{2r(t)},e^{-2r(t)})$ and $\frac12\operatorname{diag}(e^{-2r(t)},e^{2r(t)})$ along with appropriate displacement by $(m\sqrt\pi,n\sqrt\pi)$ where $m,n\in\mathbb Z$, such that all have common squeezing $r(t)$. Notably, set-magic has direct applications in communication protocols, such as enhancing the performance of quantum random access codes~\cite{Zamora2025}. An ensemble that has a significant violation combines two highly squeezed, near-vertex stabilizer states with a vacuum state displaced by $\sqrt\pi$ along one quadrature, whose logical image is a noisy magic state~(see SM~\ref{sec:finiteq_Squeezing}).

Furthermore, the ability of Gaussian states to violate the classical bound of the $S_3(t)$ witness allows for the certification of randomness, which is a crucial primitive in cryptographic protocols. We consider the generation of a binary outcome from preparation $\vec{r}_{x=2}$ and measurement $M_{y=1}$, and quantify the certified randomness by the conditional min-entropy $H_{\min}(B|E)=-\log_2 P_{\mathrm{guess}}(B|E)$, where $P_{\mathrm{guess}}(B|E)$ denotes the probability of Eve guessing correctly Bob's outcome. Following the semidefinite-programming approach of Ref.~\cite{sarubi2026}, we bound Eve's optimal guessing probability using an adapted Navascués--Vértesi hierarchy~\cite{navascues2015}. We fix the six correlators $E_{xy} = \vec{r}_x \cdot \vec{q}_y$, which represent the expectation values of the dichotomic measurement $y$ (with Bloch vector $\vec{q}_y$) performed on the prepared state $x$, and maximize Eve's guessing probability over moment matrices compatible with these correlations and the prepare-and-measure constraints.   This provides a certified lower bound on the randomness of the generated bit within the assumed prepare-and-measure model. Since the value $S_3(t)$ does not uniquely determine the correlator associated with the generation setting, we determine the minimum $H_{\min}$ among the physical Gaussian configurations that attain the maximal witness value. See SM.~\ref{sec:SM_randomness} for a detailed discussion. Applying this framework, Fig.~\hyperref[F-certification]{\ref{F-certification}(b)} shows the resulting worst-case certified min-entropy as a function of $t$ for different squeezing bounds $r \leq r_0$. At moderate squeezing ($r_0 \leq 0.262$, for example), no randomness is certified. As the squeezing bound increases, the witness exceeds the classical bound and certifies nonzero randomness. Interestingly, the certified min-entropy saturates at approximately $0.18$ bits for $r_0 \gtrsim 1.5$. Thus, the SSD mapping of Gaussian states not only enables violations of the classical and stabilizer PAM bounds, but also allows for certified randomness generation.

\textit{Discussion---}In this work, we derive an exact parametrization of the stabilizer subsystem decomposition, mapping arbitrary single-mode Gaussian states into a nontrivial region of the logical qubit Bloch sphere. By exactly characterizing the $XZ$ projection, we establish its convexity and identify the pure, axis-aligned states that dictate its boundary, also showing that the vacuum has the highest robustness of magic within this projection. 

Beyond our analytical characterization, our framework yields practical certification tools. It provides a tight squeezing criterion for Gaussian preparations. Moving from single states to ensembles, we proved that the logical image of Gaussian states with bounded squeezing can violate both stabilizer and classical bounds of PAM scenarios. Violating the stabilizer bound certifies ``set-magic"~\cite{Wagner2025}, while breaking classical bounds enables the certification of cryptographic randomness. More broadly, these results reinforce the utility of the GKP framework for accessing nonclassical correlations, complementing recent work on Bell nonlocality with GKP states~\cite{Yang2026,lopeteguigonzález2026,salek2026}.



Finally, our results clarify the relationship between Gaussianity and logical non-stabilizerness. The emergence of non-stabilizerness post-SSD should not be misinterpreted as the generation of a DV resource from a ``free" CV state via a free operation. Because the SSD map is inherently non-free within the resource theory of non-Gaussianity and more generally within the theory of Wigner Negativity \cite{rthwig}, what it actually exposes is the rich, underlying resource geometry induced by GKP extraction. Furthermore, because our exact boundary characterization focuses on the $XZ$ projection, it naturally isolates the real part of the logical qubit. This leaves the generation of logical imaginarity~\cite{Wu_2021} as a distinct operational resource to be explored. A natural future direction is to extend this geometric framework to non-Gaussian states. Investigating how specific non-Gaussian resources, such as photon subtraction or non-Gaussian noise, expand the accessible region in the Bloch sphere could provide a unified understanding of how continuous-variable non-Gaussianity translates into discrete-variable non-stabilizerness. Furthermore, it would be interesting to determine how this geometry changes under multimode extensions, finite-energy GKP code states, and realistic extraction noise, and whether the operational advantages of bounded squeezing persist in experimentally implementable settings.

\let\oldaddcontentsline\addcontentsline
\renewcommand{\addcontentsline}[3]{}
\textit{Acknowledgments---}We acknowledge financial support from the Simons Foundation (Grant No. 1023171, R.C.), the Brazilian National Council for Scientific and Technological Development (CNPq; Grants No. 403181/2024-0, 315081/2025-2 and 301687/2025-0), the EU Horizon Europe programme (QSNP, Grant No. 101114043; CLUSTEC, Grant No. 101080173; and ClusterQ, Grant No. 101055224), the National Institute of Science and Technology for Applied Quantum Computing through CNPq Process No. 408884/2024-0, and the Financiadora de Estudos e Projetos (Grant No. 1699/24 IIF-FINEP). R.C. thanks the Technical University of Denmark for its hospitality, where part of this work was carried out during a guest professorship supported by the Otto M\o nsted Foundation. S.Z. acknowledges the Coordenação de Aperfeiçoamento de Pessoal de Nível Superior -- Brasil (CAPES) -- Finance Code 001 and the support of the Natural Sciences and Engineering Research Council of Canada (NSERC), (CIRTA-P-614385-2026)

\textit{Code Availability---}The code used to generate the plots is available at~\href{https://gitlab.com/data_availability/gaussian_states_under_ssd}{Gitlab}.

\bibliographystyle{apsrev4-2}
\bibliography{2-references}
\clearpage


\onecolumngrid
\begin{center}
    \Large \bfseries {Supplementary Material}
\end{center}
\vspace{1em}

\let\addcontentsline\oldaddcontentsline
\begingroup
\parskip=0pt
\setcounter{tocdepth}{3}
\tableofcontents
\endgroup

\onecolumngrid

\section{Background}
\subsection{Continuous variable systems and GKP codes}
\subsubsection{Conventions: The Weyl operator and the characteristic function}
Throughout, $\qh,\ph$ are the canonical position and momentum operators of a single bosonic mode, satisfying $[\qh,\ph]=i$ ($\hbar=1$). We define the Weyl displacement operator as $\Wop{u}{v} \equiv e^{i(u\qh+v\ph)}$, for $(u,v)\in\mathbb R^2$. They obey the composition rule 
 \begin{equation}
     \Wop{a}{b}\Wop{c}{d} = \Wop{a+c}{b+d} e^{-\frac{i}{2}(ad-bc)},\label{eq:Composition}
 \end{equation}
 and the conjugation rule,
 \begin{equation}
     \Wop{a}{b}\Wop{c}{d}\Wop{a}{b}^{-1} = \Wop{c}{d} e^{-i(ad-bc)},\label{eq:Conjugation}
 \end{equation}
for all $a,b,c,d \in\mathbb{R}$.

Both rules follow from the Baker--Campbell--Hausdorff formula for operators,  and the fact that $\Wop{a}{b}^{-1}=\Wop{a}{b}^\dagger=\Wop{-a}{-b}$.  This can be shown as follows.
First, for Eq.~\eqref{eq:Composition}, let $A=i(a\qh+b\ph)$ and $B=i(c\qh+d\ph)$. Since $[\qh,\ph]=i$,
\begin{equation}
[A,B]= i^2\big(ad[\qh,\ph]+bc[\ph,\qh]\big)  = -i(ad-bc),
\end{equation}
which is a $c$-number, hence commutes with both $A$ and $B$. The Baker--Campbell--Hausdorff formula for operators, $e^Ae^B=e^{A+B}e^{\frac12[A,B]}$, then gives $e^Ae^B = e^{i[(a+c)\qh+(b+d)\ph]} e^{-\frac{i}{2}(ad-bc)}= \Wop{a+c}{b+d}e^{-\frac{i}{2}(ad-bc)}$, therefore $\Wop{a}{b}\Wop{c}{d} = \Wop{a+c}{b+d}e^{-\frac{i}{2}(ad-bc)}$.

For Eq.~\eqref{eq:Conjugation}, we first note that since $\qh,\ph$ are Hermitian, $\Wop{a}{b}^\dagger=e^{-i(a\qh+b\ph)}=\Wop{-a}{-b}$. Furthermore, using the composition law,   $\Wop{a}{b}\Wop{-a}{-b}=\Wop{0}{0}e^{-\frac i2(a(-b)-b(-a))}=\hat I\cdot e^0=\hat I$. Therefore the operators satisfy $\Wop{a}{b}^{-1}=\Wop{a}{b}^\dagger=\Wop{-a}{-b}$. Using this and  by applying the composition law twice we get
\begin{align*}
\Wop{a}{b}\Wop{c}{d}\Wop{-a}{-b} &= \Wop{a+c}{b+d}e^{-\frac i2(ad-bc)}\Wop{-a}{-b}\\
&= \Wop{c}{d}e^{-\frac i2\big[(a+c)(-b)-(b+d)(-a)\big]}e^{-\frac i2(ad-bc)}\\
&= \Wop{c}{d}e^{-\frac i2(da-cb)}e^{-\frac i2(ad-bc)} \\ &= \Wop{c}{d}e^{-i(ad-bc)},
\end{align*}
obtaining Eq.~\eqref{eq:Conjugation}.
 With this convention then, for a density operator $\hat\rho$, its (symmetrically-ordered) characteristic function is 
 \begin{equation}
     \chi_\rho(u,v)\equiv\Tr\big(\hat\rho\Wop{u}{v}\big).
 \end{equation}
In particular, for Gaussian states, the characteristic function is itself Gaussian. Following our convention any single mode Gaussian state has a characteristic function of the form 
\begin{equation}
    \chi_G(\vec{\xi}) = \exp\left( i\vec{\xi}^{\mathrm T}\vec{r} -\frac12 \vec{\xi}^{\mathrm T}V\vec{\xi}\right),\label{eq:Gaussian_char}
\end{equation}
where $\vec{\xi} = (u,v)^{\intercal}\in \mathbb R^2$, $\vec{r} = (q_0, p_0)^{\intercal} = (\langle \qh\rangle,\langle \ph\rangle)^{\intercal}$, and $V$ the covariance matrix
\begin{equation}
V=\begin{pmatrix}
    V_{qq} & V_{qp}\\
    V_{pq} & V_{pp}
  \end{pmatrix},
\end{equation}
with entries
\begin{align}
V_{qq}&= \left\langle (\Delta\qh)^2\right\rangle,\nonumber\\
V_{pp}&= \left\langle (\Delta\ph)^2\right\rangle,\nonumber\\
V_{qp}&= V_{pq}=\frac12\left\langle\Delta\qh\Delta\ph+\Delta\ph\Delta\qh\right\rangle,\nonumber
\end{align}
where $\Delta\qh=\qh-q_0$, and similarly for $\ph$. The covariance matrix satisfies the Schr\"odinger--Robertson uncertainty relation $\det V\geq\frac{1}{4}$, where $\det V=V_{qq}V_{pp}-V_{qp}^2$ is the determinant of $V$. For a general single-mode Gaussian state, $V$ can be written as \mbox{$V=\nu R(\theta)S(2r)R(\theta)^T$}, where $\nu=\sqrt{\det V}\geq\tfrac12$ is the symplectic eigenvalue of $V$. The squeezing $\hat S(r)$ and the rotation $\hat R(\theta)$ matrices are given by
\begin{equation}
    S(r) =\begin{pmatrix}
        e^{-r}&0\\
        0&e^{r}
    \end{pmatrix},\quad R(\theta) =\begin{pmatrix}
        \cos(\theta)&-\sin(\theta)\\
        \sin(\theta)&\cos(\theta)
    \end{pmatrix},
\end{equation}
which implies that for pure states, the covariance matrix is parameterized by~\cite{Weedbrook_2012}
\begin{align}
    &V_{qq} = \frac{1}{2}(\cosh(2r)-\cos(2\theta)\sinh(2r))\nonumber\\
    &V_{pp} = \frac{1}{2}(\cosh(2r)+\cos(2\theta)\sinh(2r))\label{eq:V_parameterized}\\
    &V_{qp} = -\frac{1}{2}\sin(2\theta)\sinh(2r).\nonumber
\end{align}
 \subsubsection{The GKP code}
 Define the ideal GKP code state as
 \begin{equation}
     |l_\GKP\rangle \equiv \sum_{n\in\mathbb Z} \big|\qh=(2n+l)\sqrt\pi\big\rangle, \label{eq:GKP_state}
 \end{equation}
for $l\in\{0,1\}$. The square GKP stabilizer lattice has spacing $2\sqrt\pi$ in both Weyl directions, the logical Pauli displacements have spacing $\sqrt\pi$, and the syndrome cell used below is $[-\sqrt\pi/2,\sqrt\pi/2)^2$. In the ideal case, $|l_\GKP\rangle$ is not a normalizable Hilbert-space vector; this is the standard idealization in the GKP literature~\cite{GKP2001,Terhal2016,Grimsmo2021}. For $m,n\in\mathbb Z$, define the stabilizers $\Sop{m}{n}\equiv\Wop{2m\sqrt\pi}{2n\sqrt\pi}$. The $\ph$ part translates the comb by $2n\sqrt\pi$, while the $\qh$ part gives the phase $e^{i2\pi m(2k+l)}=1$. The composition rule~\eqref{eq:Composition} then gives $\Sop{m}{n}\Sop{m'}{n'}=\Sop{m+m'}{n+n'}$. Similarly, we define the specific displacement operators~\cite{GKP_2001}
 \begin{align}
&\hat Z \equiv \Wop{\sqrt\pi}{0}=e^{i\sqrt\pi\qh},\\
&\hat X \equiv \Wop{0}{-\sqrt\pi}=e^{-i\sqrt\pi\ph},\\
&\hat Y \equiv i\hat X\hat Z = \Wop{\sqrt\pi}{-\sqrt\pi}.
 \end{align}
One can see that applying $\hat Z$ to Eq~\eqref{eq:GKP_state}, one obtains a phase of $e^{i(2n+l)\pi} = (-1)^l$ while applying $\hat X$ produces a bit flip. $\hat Y$ follows from applying the composition law to the product $\hat X\hat Z$. In particular it is such that $\hat Y|0_\GKP\rangle=i|1_\GKP\rangle$, $\hat Y|1_\GKP\rangle=-i|0_\GKP\rangle$.  While these operators act on the full infinite-dimensional space, they implement the standard qubit Pauli operations when restricted to the code space $\mathrm{span}\{|0_\GKP\rangle,|1_\GKP\rangle\}$. This restriction yields the $2\times2$ logical operators $\hat Z_L = \hat\Pi\hat Z\hat\Pi$, $\hat X_L = \hat\Pi\hat X\hat\Pi$, and $\hat Y_L = \hat\Pi\hat Y\hat\Pi$ used in our Bloch vector decomposition.

Finally, we introduce the code projector 
\begin{equation}
    \Pih = \frac1{\sqrt\pi}\sum_{m,n\in\mathbb Z}\Sop{m}{n}.\label{eq:Projectors_Stabilizers}
\end{equation}
This is the Poisson-summation representation of the distributional code-space kernel $\Pih\equiv|0_\GKP\rangle\langle0_\GKP|+|1_\GKP\rangle\langle1_\GKP|$ in terms of the code's stabilizer operators. The proof of the above identity is as follows:
\begin{proposition}[Stabilizer representation of $\Pih$]\label{thm:poisson}
\begin{equation}
\Pih = \frac1{\sqrt\pi}\sum_{m,n\in\mathbb Z}\Sop{m}{n}.    
\end{equation}
\end{proposition}
\begin{proof}
We compare the position-space matrix elements $\langle x|\Pih|x'\rangle$ on both sides.

\emph{Left side.} Write $\mathrm{comb}_l(x)=\sum_n\delta\big(x-(2n+l)\sqrt\pi\big)$, so $\langle x|\Pih|x'\rangle=\sum_{l=0,1}\mathrm{comb}_l(x)\mathrm{comb}_l(x')$. By the classical Poisson summation formula (see e.g.\ \cite{SteinShakarchi}), a Dirac comb of period $a$ offset by $x_0$ has Fourier series
\begin{equation}
\sum_{n\in\mathbb Z}\delta(x-na-x_0)=\frac1a\sum_{k\in\mathbb Z}e^{2\pi ik(x-x_0)/a}.    
\end{equation}

With $a=2\sqrt\pi$, $x_0=l\sqrt\pi$, this gives $\mathrm{comb}_l(x)=\frac1{2\sqrt\pi}\sum_k(-1)^{kl}e^{ik\sqrt\pi x}$. Hence
\begin{align}
\langle x|\Pih|x'\rangle &= \frac1{4\pi}\sum_{k,k'}e^{ik\sqrt\pi x}e^{ik'\sqrt\pi x'}\sum_{l=0,1}(-1)^{(k+k')l} \nonumber\\
&=\frac1{2\pi}\sum_{\substack{k,k'\\ k+k'\ \mathrm{even}}} e^{ik\sqrt\pi x}e^{ik'\sqrt\pi x'},
\end{align}
using $\sum_{l=0,1}(-1)^{(k+k')l}=2$ if $k+k'$ even, $0$ otherwise. Splitting into ``both even'' ($k=2p,k'=2q$) and ``both odd'' ($k=2p+1,k'=2q+1$) pieces,
\begin{align}
&\langle x|\Pih|x'\rangle= \nonumber\\
&\frac1{2\pi}\Big[\sum_{p,q}e^{i2p\sqrt\pi x}e^{i2q\sqrt\pi x'}+\sum_{p,q}e^{i(2p+1)\sqrt\pi x}e^{i(2q+1)\sqrt\pi x'}\Big]. \label{eq:LHS}
\end{align}

\emph{Right side.} Using $\Wop{a}{b}|x'\rangle=e^{iab/2}e^{ia(x'-b)}|x'-b\rangle$ (consider $\Wop{a}{b}=e^{ia\qh}e^{ib\ph}e^{iab/2}$), with $(a,b)=(2m\sqrt\pi,2n\sqrt\pi)$,
\begin{equation}
\langle x|\Sop{m}{n}|x'\rangle = e^{i2m\sqrt\pi x'}\delta(x-x'+2n\sqrt\pi)    
\end{equation}

(the phase $e^{i\cdot2m\sqrt\pi\cdot2n\sqrt\pi/2}e^{-i2m\sqrt\pi\cdot2n\sqrt\pi}=e^{-2\pi imn}=1$). Summing over $n$ produces another Dirac comb, now in the variable $y=x-x'$, with period $2\sqrt\pi$. By Poisson summation again, $\sum_n\delta(y+2n\sqrt\pi)=\frac1{2\sqrt\pi}\sum_k e^{ik\sqrt\pi y}$. Therefore
\begin{align}
&\frac1{\sqrt\pi}\sum_{m,n}\langle x|\Sop{m}{n}|x'\rangle\nonumber\\
&=\frac1{2\pi}\sum_{m,k}e^{i2m\sqrt\pi x'}e^{ik\sqrt\pi(x-x')}\nonumber\\
&=\frac1{2\pi}\sum_{m,k}e^{ik\sqrt\pi x}e^{i(2m-k)\sqrt\pi x'}.
\end{align}
For fixed $k$, as $m$ ranges over $\mathbb Z$, $j\equiv 2m-k$ ranges over all integers of the same parity as $k$. Splitting the outer sum by the parity of $k$ and relabeling ($k=2p$, $j=2q$ for even and $k=2p+1$, $j=2q+1$ for odd) gives
\begin{align}
&\frac1{\sqrt\pi}\sum_{m,n}\langle x|\Sop{m}{n}|x'\rangle = \frac1{2\pi}\Big[\sum_{p,q}e^{i2p\sqrt\pi x}e^{i2q\sqrt\pi x'}+\sum_{p,q}e^{i(2p+1)\sqrt\pi x}e^{i(2q+1)\sqrt\pi x'}\Big]. \label{eq:RHS}
\end{align}
Equations~\eqref{eq:LHS} and \eqref{eq:RHS} are  identical.
\end{proof}

An important remark is that formally expanding $\Pih^2=\frac1\pi\sum_{m,n,m',n'}\Sop{m}{n}\Sop{m'}{n'}$ and using the Abelian group law, every target stabilizer $\Sop{M}{N}$ receives contributions from \emph{infinitely many} pairs $(m,n),(m',n')$ with $m+m'=M,n+n'=N$, so $\Pih^2$ diverges rather than reproducing $\Pih$. This is expected since $\Pih$, like $|x\rangle\langle x|$, is a formal/distributional object, not a bounded operator on the Hilbert space, and $\Pih^2$ is simply not a well-posed expression on its own. It only needs to be well-defined when integrated against a trace-class $\hat\rho$ inside the syndrome integral of definition~\eqref{eq:SSDmap} below. 
\subsubsection{The stabilizer subsystem decomposition}\label{sec:SSD}
For a family of unitaries $\hat D(\vec{s})$, $\vec{s}=(s_q,s_p)$, representing displacement by the syndrome $\vec{s}$, Ref.~\cite{Calcluth_2024} defines the stabilizer subsystem decomposition (SSD) as 
\begin{equation}
\hat\rho_L= \Lambda(\hat\rho) \equiv \frac1{\sqrt\pi}\int_{-\sqrt\pi/2}^{\sqrt\pi/2} ds_q\int_{-\sqrt\pi/2}^{\sqrt\pi/2} ds_p\ \hat\Pi\hat D(-\vec{s})\hat\rho\hat D(-\vec{s})^\dagger\hat\Pi. \label{eq:SSDmap}
\end{equation}
This equation constitutes a map from any continuous-variable state $\hat\rho$ to a discrete logical state $\hat\rho_L= \Lambda(\hat\rho)$. Now consider evaluating the expectation value of a logical operator $\hat O_L = \hat\Pi\hat O\hat\Pi$, where $\hat O \in \{\hat X, \hat Y, \hat Z\}$ is the corresponding continuous-variable displacement operator. Since $\hat O$ commutes with the stabilizers, it commutes with the code projector, yielding $\hat\Pi\hat O\hat\Pi = \hat O\hat\Pi^2$. Note that squaring the ideal projector produces a divergence, $\hat\Pi^2 \propto \delta(0)\hat\Pi$. However, as established in the previous subsection, this expression is well-defined when integrated within the SSD map. Therefore, under the syndrome integrals, one can write
\begin{align}
\Tr\big[\hat\Pi\hat D(-\vec{s})\hat\rho\hat D(-\vec{s})^\dagger\hat\Pi\hat O_L\big]&=\Tr\big[\hat D(-\vec{s})\hat\rho\hat D(-\vec{s})^\dagger\hat\Pi\hat O_L\hat\Pi\big]\nonumber\\
&=\Tr\big[\hat D(-\vec{s})\hat\rho\hat D(-\vec{s})^\dagger\hat O\hat\Pi^2\big]\nonumber\\
&=\Tr\big[\hat D(-\vec{s})\hat\rho\hat D(-\vec{s})^\dagger\hat O\hat\Pi\big],
\end{align}
so we have 
\begin{equation}
    \Tr\big(\hat\rho_L\hat O_L\big) = \frac1{\sqrt\pi}\int\int d\vec{s}\ \Tr\Big[\hat\rho\hat D(-\vec{s})^\dagger\big(\hat O\hat\Pi\big)\hat D(-\vec{s})\Big].\label{eq:OL_avg}
\end{equation}
We now note that Eq~\eqref{eq:OL_avg} is invariant under any of the following conventions for $\Dh(s)$,
\begin{enumerate}[label=(\roman*)]
\item $\Dh_{\rm A}(\vec{s})=e^{i(s_q\qh+s_p\ph)} =\Wop{s_q}{s_p}$
\item $\Dh_{\rm B}(\vec{s})=e^{i(s_q\qh-s_p\ph)}=\Wop{s_q}{-s_p}$ 
\item $\Dh_{\rm C}(\vec{s})=e^{is_q\qh}e^{-is_p\ph}$ 
\end{enumerate}
First, using the composition law~\eqref{eq:Composition} we see that $\Dh_{\rm C}(\vec{s})=e^{is_q\qh}e^{-is_p\ph}=\Wop{s_q}{-s_p}e^{is_qs_p/2}$, \textit{i.e.} $\Dh_{\rm C}(\vec{s})=\Dh_{\rm B}(s)e^{is_qs_p/2}$: a pure phase times $\Dh_{\rm B}$.  Since a scalar phase commutes with every operator, it cancels identically in any conjugation $\Dh(-\vec{s})^\dagger(\cdot)\Dh(-s)$ appearing in Eq~\eqref{eq:OL_avg}: writing $\Dh_{\rm C}(-\vec{s})=\Dh_{\rm B}(-\vec{s})\lambda(\vec{s})$ for the scalar $\lambda(\vec{s})=e^{is_qs_p/2}$ (unchanged under $\vec{s}\to-\vec{s}$, since it depends only on the product $s_qs_p$), we get $\Dh_{\rm C}(-\vec{s})^\dagger \hat O\Dh_{\rm C}(-\vec{s})=\lambda^*(\vec{s})\Dh_{\rm B}(-\vec{s})^\dagger \hat O\lambda(\vec{s})\Dh_{\rm B}(-\vec{s})=\Dh_{\rm B}(-\vec{s})^\dagger\hat O\Dh_{\rm B}(-\vec{s})$ for any operator $\hat O$. So (iii) reduces exactly to (ii). 

For (i), $\Dh_{\rm A}(-s)^\dagger=\Wop{s_q}{s_p}$; for (ii), $\Dh_{\rm B}(-\vec{s})^\dagger=\Wop{s_q}{-s_p}$. Conjugating $\Wop{c}{d}$ by $\Wop{a}{b}$ produces the phase $e^{i(bc-ad)}$; substituting $(a,b)=(s_q,\pm s_p)$ and considering $\Wop{c}{d}$ the associated operators to $\hat O\Pih$, then integrating over $s_p\in(-\sqrt\pi/2,\sqrt\pi/2)$ gives, respectively, $\int e^{is_pc}ds_p=K(c)$ and $\int e^{-is_pc}ds_p=K(-c)$. But note that  $K$ is even: $K(w)=\frac{2\sin(w\sqrt\pi/2)}{w}=K(-w)$. Hence $K(c)=K(-c)$, and the two conventions give identical integrals.

Therefore we can fix $\Dh(\vec{s})=\Dh_{\rm A}(\vec{s})=\Wop{s_q}{s_p}$.

\subsection{Discrete variable quantities}
\subsubsection{Robustness of Magic}
 Magic measures provide a useful framework for quantifying the non-stabilizerness, and hence the resourcefulness, of discrete-variable quantum states. One such measure is the robustness of magic (RoM). To define it, let $\mathcal{S}_n$ denote the set of pure stabilizer states on $n$ qubits. Since stabilizer states span the space of quantum states, an arbitrary state $\hat{\rho}$ can be written as a linear combination of elements of $\mathcal{S}_n$, i.e., $\hat{\rho} = \sum_i x_i \hat{\sigma}_i$ with $\hat{\sigma}_i \in \mathcal{S}_n$, where the coefficients $x_i$ are real and need not be positive. Such a decomposition is generally not unique, and different sets of coefficients ${x_i}$ may represent the same state. The RoM quantifies the minimum total weight required in such a decomposition, and is therefore defined as the minimum $\ell_1$-norm of the coefficients~\cite{Calcluth_2024}:
\begin{equation}
R(\hat{\rho}) =\min_{\{x_i\}}\left(\sum_i |x_i|:\hat{\rho} = \sum_i x_i\hat{\sigma}_i\right).
\label{eq:ROM}
\end{equation}
Every state in the stabilizer polytope, $\operatorname{conv}(\mathcal{S}_n)$, has a convex decomposition into pure stabilizer states with total coefficient weight one. Together with trace normalization, this gives $R(\hat\rho)=1$ for every stabilizer state. Thus, $R(\hat\rho)>1$ signals non-stabilizerness.

For a single qubit, the RoM admits a particularly simple expression in terms of the state's Bloch-vector components. Specifically,
\begin{equation}
R^{(1)}(\hat{\rho})=\max\biggl\{1,\left|\Tr(\hat{\rho}\hat{X})\right|+\left|\Tr(\hat{\rho}\hat{Y})\right|+\left|\Tr(\hat{\rho}\hat{Z})\right|\biggl\},
\label{eq:single_qubit_ROM}
\end{equation}
where here we are using $\hat{X}$, $\hat{Y}$, and $\hat{Z}$ to denote the Pauli operators. This expression directly relates the RoM to the components of the Bloch vector and provides a convenient way of evaluating the amount of magic present in a single-qubit state.
\subsubsection{Prepare-and-measure witnesses}\label{sec:pam_scenario}
The prepare-and-measure scenario is constituted by a party ($A)$ that prepares and sends physical systems to another party ($B$) which measure states. Each party performs its duty based on a random variable ($X$ for $A$, $Y$ for $B$): for example $A$ prepares the state $\rho_{X=x}$  while $B$ measures the observable $M_{Y=y}$. Similarly to  other experimental scenarios (such as the Bell scenario), the prepare-and-measure scenario can be described through the lens of causal modeling and it also consitutes a platform for the study of semi-device-independent protocols. However, in this work we are interested in the purely geometrical aspect of the witnesses that it provides and in principle we disregard the possible device-independent certification properties. In this work we focus in the inequality~\cite{Taoutioui2025},which has been of interest in the study of magic~\cite{Zamora2025}.  
\begin{equation}
    S_3(t) = 2t(E_{11}+E_{21}-E_{31}) + 2(1-t)(E_{12}-E_{22}),\label{eq: tilted_s3}
\end{equation} 
defined by $t\in[0,1]$. Here $E_{xy}=p(0|xy)-p(1|xy)$ and $p(b|xy)$ is the probability of obtaining the outcome $b$, given $B$ measured the observable $y$ of the states $\rho_x$. In the same reference, the authors obtained the quantum value (the maximum value possible using quantum theory)
    \begin{equation}
        S_3^Q(t) = 2t + 4\sqrt{t^2 + (1-t)^2},
    \end{equation}
and the  classical value (the maximum value possible using classical probabilistic theory)
\begin{equation}
    S_3^C(t) =
    \begin{cases}
        4-2t\quad &\text{if}\quad t <1/2\\
        6t \quad  &\text{if}\quad 1/2\leq t
    \end{cases}.  
\end{equation}
More recently, in Ref.~\cite{Zamora2025}, the inequality was used to obtain a  stabilizer bound(the maximum value possible using stabilizer states and arbitrary measurements) for this inequality
\begin{equation}
        S_3^{\mathrm{STAB}}(t) = \max\left(6t, \ 4-2t, \ 2t\sqrt{5} + 2(1-t)\sqrt{2}\right),\label{eq:Stab_bound_S3}
    \end{equation}
When considering quantum measurements, after optimizing analytically over them, it can be shown that $S_3(t)$ may be written in terms of only the Bloch vectors of the three states giving the expression 
\begin{equation}
 S_3(t) = 2t\|\vec{r}_1 + \vec{r}_2 -\vec{r}_3\|_2  + 2(1-t)\|\vec{r}_1-\vec{r}_2\|_2.\label{S3_tilted_r}
\end{equation} 
Although we will use this expression for our results, we emphasize that  considering the SSD map as a method for  the preparation of the DV states sent in a prepare-and-measure scenario constitutes a very strong assumption, imposing an obstacle to a SDI certification of CV systems in this approach.

\section{Results}
We next derive a lattice expression for $\Tr_L(\hat\rho_L\hat O_L)$ when the series is absolutely convergent, in particular for every finite-covariance Gaussian input considered here. The operator $\hat O_L$ is the ordinary logical matrix induced by a continuous-variable displacement $\hat O$ that commutes with the stabilizers. For more general inputs, a weak or regularized summation prescription is required.

\subsection{The Bloch vector of $\hat\rho_L$}\label{sec:Bloch_vec_ssd}
First, we state the following lemma.

\begin{lemma}[Kernel function]\label{lem:kernel}
Define $K(w)\equiv\displaystyle\int_{-\sqrt\pi/2}^{\sqrt\pi/2}e^{iws}ds$. Then, for $k\in\mathbb Z$,
\begin{align}
&K(2k\sqrt\pi)=\sqrt\pi\delta_{k,0}\\
&K\big((2k+1)\sqrt\pi\big)=\frac{2(-1)^k}{(2k+1)\sqrt\pi}.
\end{align}
\end{lemma}
\begin{proof}
Direct evaluation of the integral at $w=0$ yields $K(0) = \sqrt{\pi}$. For any non-zero real $w$, the integral evaluates to
\begin{equation}
    K(w) = \left[ \frac{e^{iws}}{iw} \right]_{-\sqrt\pi/2}^{\sqrt\pi/2} = \frac{2\sin(w\sqrt\pi/2)}{w}.
\end{equation}
Consider $w = 2k\sqrt{\pi}$ for an integer $k \neq 0$. The numerator becomes $\sin(k\pi)=0$, yielding $K(2k\sqrt{\pi}) = 0$. Combining this with the $k=0$ case establishes the first identity, 
\begin{equation}
    K(2k\sqrt\pi) = \sqrt\pi\delta_{k,0}.
\end{equation}
For odd multiples of $\sqrt{\pi}$, we substitute $w = (2k+1)\sqrt{\pi}$ into the general formula. The argument of the sine function becomes $(2k+1)\pi/2 = k\pi + \pi/2$. This implies $\sin(k\pi + \pi/2) = \cos(k\pi) = (-1)^k$, giving the second identity,
\begin{equation}
    K\big((2k+1)\sqrt\pi\big) = \frac{2(-1)^k}{(2k+1)\sqrt\pi}.
\end{equation}
\end{proof}
Now consider that for $\hat O=\Wop{u_0}{v_0}$,
\begin{equation}
    \hat O\Sop{m}{n} = \Phi_{m,n}\Wop{u_0+2m\sqrt\pi}{v_0+2n\sqrt\pi}\label{eq:prodOLSmn},
\end{equation}
with $\Phi_{m,n}\equiv e^{-i\sqrt\pi(nu_0-mv_0)}$. This follows from the Composition law~\eqref{eq:Composition} with $(a,b,c,d)=(u_0,v_0,2m\sqrt\pi,2n\sqrt\pi)$: the phase is $e^{-\frac i2(u_0\cdot2n\sqrt\pi-v_0\cdot2m\sqrt\pi)}=e^{-i\sqrt\pi(nu_0-mv_0)}$.

Under these considerations we can state the following result
\begin{proposition}[Expectation value of $\hat O_L$]\label{thm:master}
For a continuous-variable operator $\hat O=\Wop{u_0}{v_0}$ satisfying $[\hat\Pi, \hat O] = 0$, the expectation value of its logical counterpart $\hat O_L = \hat\Pi\hat O\hat\Pi$ is given by
\begin{align}
    \Tr(\hat\rho_L\hat O_L)=\frac1\pi\sum_{m,n=-\infty}^{\infty}&\Phi_{m,n}K(u_0+2m\sqrt\pi)K(v_0+2n\sqrt\pi) \chi_\rho(u_0+2m\sqrt\pi,v_0+2n\sqrt\pi). \label{eq:MasterFormula}
\end{align}
\end{proposition}
\begin{proof}
Substitute $\Pih=\frac1{\sqrt\pi}\sum_{m,n}\Sop{m}{n}$ defined in Eq~\eqref{eq:Projectors_Stabilizers} into Eq~\eqref{eq:OL_avg}:
\[
\Tr(\hat\rho_L\hat O_L)=\frac1\pi\sum_{m,n}\int\int d\vec{s}\ \Tr\big[\hat\rho\Dh(\vec{s})\big(\hat O\Sop{m}{n}\big)\Dh(-\vec{s})\big],
\]
using $\Dh(-\vec{s})^\dagger=\Dh(\vec{s}) = \Wop{s_q}{s_p}$ (convention (i)). Using Eq~\eqref{eq:prodOLSmn}, $\hat O\Sop{m}{n}=\Phi_{m,n}\Wop{c}{d}$ with $(c,d)=(u_0+2m\sqrt\pi,v_0+2n\sqrt\pi)$. By the conjugation rule~\eqref{eq:Conjugation}, $\Dh(\vec{s})\Wop{c}{d}\Dh(-\vec{s})=\Wop{c}{d}e^{i(s_pc-s_qd)}$, so
\[
\Tr\big[\hat\rho\Dh(\vec{s})\Wop{c}{d}\Dh(-\vec{s})\big]=e^{i(s_pc-s_qd)}\chi_\rho(c,d).
\]
Integrating over the (independent, factorized) domain gives $\int ds_pe^{is_pc}\int ds_qe^{-is_qd}=K(c)K(d)$ by Lemma~\ref{lem:kernel} (using $K$ even for the $s_q$ integral). Assembling all factors gives the claim.
\end{proof}
Eq.~\eqref{eq:MasterFormula} gives a general expression for any logical operator generated by a CV operator that commutes with the projector. In particular we are interested in the  Bloch vector of the resulting state. Therefore, by inserting $\hat O = \hat Z,\hat X,\hat Y$ in Eq.~\eqref{eq:MasterFormula}, we obtain the Bloch vector $\vec R_L =(X,Y,Z)$ of $\hat\rho_L$ with components defined through the following expressions
\begin{align}
    &Z=\langle\ZL\rangle = \frac2\pi\sum_{m=-\infty}^{\infty}\frac{(-1)^m}{2m+1}\chi_\rho\big((2m+1)\sqrt\pi,0\big),\label{eq:Z}\\
    &X=\langle\XL\rangle = \frac2\pi\sum_{m=-\infty}^{\infty}\frac{(-1)^m}{2m+1}\chi_\rho\big(0,(2m+1)\sqrt\pi\big),\label{eq:X}\\
    &Y=\langle\YL\rangle = -\frac4{\pi^2}\sum_{m,n=-\infty}^{\infty}\frac{\chi_\rho\big((2m+1)\sqrt\pi,(2n+1)\sqrt\pi\big)}{(2m+1)(2n+1)}.\label{eq:Y}
\end{align}

\subsection{Gaussian states under the SSD map}\label{sec:Gaussian_Set}
Equations~\eqref{eq:Z},~\eqref{eq:X} and~\eqref{eq:Y} give an analytical expression for the image of all Gaussian states under the SSD map~\eqref{eq:SSDmap}. To see this, consider the characteristic function for Gaussian states given in Eq~\eqref{eq:Gaussian_char}. Expanding the vector notation we get, 
\begin{equation}
    \chi_G(u,v) = \exp\left[ i(uq_0 + vp_0) - \frac{1}{2}\big( V_{qq}u^2 + V_{pp}v^2 + 2V_{qp}uv \big) \right],
    \label{eq:Gaussian_Char_Scalar}
\end{equation}
Substituting into the Equations for the Bloch vector components we have
\begin{align}
    &Z= \frac4\pi\sum_{m=0}^{\infty}\frac{(-1)^m}{2m+1}\cos((2m+1)\sqrt{\pi}q_0)\exp\{-\frac{\pi}{2}(2m+1)^2V_{qq}\},\label{eq:Z_gaussian}\\
    &X = \frac4\pi\sum_{m=0}^{\infty}\frac{(-1)^m}{2m+1}\cos((2m+1)\sqrt{\pi}p_0)\exp\{-\frac{\pi}{2}(2m+1)^2V_{pp}\},\label{eq:X_gaussian}\\
    &Y = -\frac{8}{\pi^2}\sum_{m=0}^\infty\sum_{n=-\infty}^{\infty}\frac{\cos(\sqrt{\pi}[(2m+1)q_0 +(2n+1)p_0])}{(2m+1)(2n+1)}\times\nonumber\\
    &\hspace{2.5cm}\times \exp\{-\frac{\pi}{2}[(2m+1)^2V_{qq}+(2n+1)^2V_{pp} +2V_{pq}(2m+1)(2n+1)]\}.\label{eq:Y_gaussian}
\end{align}
These series converge for finite positive-definite covariance matrices because of the Gaussian factors. Here we show the first-term expressions, mentioned in the main text.  In particular,
\begin{align}
    &Z\simeq \frac4\pi\cos(\sqrt{\pi}q_0)\exp\{-\frac{\pi}{2}V_{qq}\},\label{eq:Z_gaussian_approx}\\
    &X \simeq \frac4\pi\cos(\sqrt{\pi}p_0)\exp\{-\frac{\pi}{2}V_{pp}\},\label{eq:X_gaussian_approx}\\
    &Y \simeq -\frac{8}{\pi^2}\big(\cos(\sqrt{\pi}[q_0 +p_0])\exp\{-\frac{\pi}{2}[V_{qq}+V_{pp}+2V_{pq}]\}\nonumber\\
    &\quad\quad-\cos(\sqrt{\pi}[q_0 -p_0])\exp\{-\frac{\pi}{2}[V_{qq}+V_{pp}-2V_{pq}]\}\big).\label{eq:Y_gaussian_approx}
\end{align}
Note that it corresponds to the $m=0$ term  and $n=0,-1$ terms for $Y$. Now, we are interested in the border of the $3D$ image set of the SSD map inside the Bloch sphere. To approximate it,  we maximize $\|(X,Y,Z)\|_2\leq1$ obtaining an outer envelope. Note that to get the maximum, we need all the cosines to be on their optima. This implies that $q_0,p_0$ must be multiples of $\sqrt{\pi}$ (note that this is true for $Y$ since the two cosines always have the same sign). Therefore, focusing on the positive quadrant $(q_0,p_0)=(0,0)$ (note that the other choices are just reflections of this case) we get the approximated Bloch vector $\vec{R}$:
\begin{equation}
    \vec{R} = \frac{4}{\pi}\big(e^{-\frac{\pi V_{pp}}{2}},\frac{4}{\pi}e^{-\frac{\pi}{2}(V_{qq}+V_{pp})}\sinh(V_{pq}\pi),e^{-\frac{\pi V_{qq}}{2}}\big).\label{eq:extremal_bv}
\end{equation}
For this zero-displacement family, the first-term expressions give
\begin{equation}
    Y \simeq XZ\sinh(V_{qp}\pi),
\end{equation}
and then we use the Schr\"odinger--Robertson uncertainty relation, $V_{qq}V_{pp}\geq1/4+V_{qp}^2$. Since we want the maximum value for the norm, we want to maximize the exponentials. Therefore we take $V_{qp}^2=V_{qq}V_{pp}-1/4$. From the approximated $X$ and $Z$ components we get finally
\begin{equation}
    Y =\pm XZ\sinh\left(\sqrt{4\ln(\frac{\pi X}{4})\ln(\frac{\pi Z}{4})-\frac{\pi^2}{4}}\right)\label{eq:App_image_XYZ}.
\end{equation}
This is a leading-order relation for zero-displacement pure inputs, not an established boundary of the full three-dimensional Gaussian image. Its $Y=0$ locus satisfies the approximated relation:
\begin{equation}
    \ln(\frac{\pi X}{4})\ln(\frac{\pi Z}{4}) =\frac{\pi^2}{16}.\label{eq:approx_ZX}
\end{equation}
The last expression can be obtained directly by considering only the $X$ and $Z$ component of the approximated Bloch vector in Eq.~\eqref{eq:extremal_bv}, which may be compared with the exact characterization shown in the next subsection.

\subsection{Exact characterization of the XZ projection}\label{sec:exact_XZ}

The approximation given in Eq.~\eqref{eq:approx_ZX} was obtained by considering only the first terms of the series in Eqs.~\eqref{eq:Z_gaussian}--\eqref{eq:X_gaussian}. In this subsection, we work with the full series to determine the exact $XZ$ projection and establish its convexity. The central simplification is that $Z$ depends on the covariance matrix only through $V_{qq}$, and $X$ only through $V_{pp}$. At fixed diagonal entries, the off-diagonal entry $V_{qp}$ does not appear in either component. When determining the reachable pairs $(X,Z)$, it enters only through the physicality constraint $V_{qq}V_{pp}-V_{qp}^2\ge\tfrac14$. This allows us to characterize the projection by combining the two quadrature variances with the displacement dependence.

\subsubsection*{Separating the variance and displacement dependence}

The expressions for $X$ and $Z$ have the same functional form. Consequently, we introduce a single function to study how each component depends on its quadrature variance and displacement. For $t>0$ and $\varphi \in \mathbb R$, define
\begin{equation}\label{eq:F_kernel}
    A(t,\varphi) := \frac{4}{\pi}\sum_{m=0}^\infty\frac{(-1)^m}{2m+1}e^{-\tfrac{\pi}{2}(2m+1)^2t}\cos[(2m+1)\varphi],
\end{equation}
so that, directly from Eqs.~\eqref{eq:Z_gaussian}-\eqref{eq:X_gaussian}, $Z=A(V_{qq}, \sqrt{\pi}q_0)$ and $X=A(V_{pp}, \sqrt{\pi}p_0)$. We also define
\begin{equation}\label{Eq:app-function-A}
    A(t):=A(t,0) = \frac{4}{\pi}\sum_{m=0}^\infty\frac{(-1)^m}{2m+1}e^{-\tfrac{\pi}{2}(2m+1)^2t},
\end{equation}
which is simply $A(t,\varphi)$ evaluated at zero displacement. Thus, $A(t)$ gives the exact value of $Z$ (or $X$) for an undisplaced Gaussian state with variance $t$ along the corresponding quadrature. Having isolated the dependence on the variance and displacement in the kernel $A(t,\varphi)$, we first determine the effect of varying the displacement while keeping the variance fixed. The following lemma shows that zero displacement maximizes the magnitude of $A(t,\varphi)$ and, more generally, characterizes its full range at fixed $t$.

\begin{lemma}[Fixed-variance extremality]\label{App:lemm-fixed-variance} For every fixed $t>0$, the function $\varphi \mapsto A(t,\varphi)$ is even, $\pi$-antiperiodic, and strictly decreasing on $(0,\pi)$. Consequently its range over $\varphi\in\mathbb{R}$ is exactly the interval $[-A(t), A(t)]$ with 
\begin{equation}\label{eq:extremal_conditions}
    A(t,\varphi) = A(t) \Longleftrightarrow \varphi \in 2\pi \mathbb{Z} \:\:\:\text{and}\:\:\: |A(t,\varphi)| = A(t) \Longleftrightarrow \varphi \in \pi \mathbb{Z}.
\end{equation}
\end{lemma}
\begin{proof}
The identities $A(t,-\varphi)=A(t,\varphi)$ and $A(t,\varphi+\pi)=-A(t,\varphi)$ follow directly from the cosine factors in Eq.~\eqref{eq:F_kernel}, since $\cos\big[(2m+1)(\varphi+\pi)\big]=-\cos\big[(2m+1)\varphi\big]$ for every integer $m$. To prove the remaining claim, we determine the sign of the derivative with respect to $\varphi$. The defining series and its differentiated series converge uniformly for $t\geq \epsilon>0$. Term-by-term differentiation gives $\partial_\varphi A(t,\varphi) = -\frac2\pi\vartheta_1(\varphi\mid 2it)$, where $\vartheta_1(u\mid\tau):=2\sum_{m\ge0}(-1)^m e^{i\pi\tau(m+1/2)^2}\sin\big[(2m+1)u\big]$ is the odd Jacobi theta function~\cite{WeissteinJacobiTheta}. Its product representation~\cite[Eq.~20.5.1]{DLMF} reads
\begin{equation}\label{eq:theta_product}
\vartheta_1(\varphi\mid2it) = 2e^{-\pi t/2}\sin\varphi\prod_{n\ge1}\Big[(1-e^{-4\pi nt})\big(1-2e^{-4\pi nt}\cos2\varphi+e^{-8\pi nt}\big)\Big].
\end{equation}
For $0<\varphi<\pi$ the prefactor $\sin\varphi$ is positive. Every factor in the product is also positive, since $1-2r\cos2\varphi+r^2=|1-re^{2i\varphi}|^2>0$ for $0<r<1$. Moreover, $\sum_n e^{-4\pi nt}<\infty$, so the product converges to a strictly positive value. Consequently, $\partial_\varphi\mathcal F(t,\varphi)<0$ on $(0,\pi)$. Thus, $\mathcal F(t,\cdot)$ is strictly decreasing on $(0,\pi)$, from $\mathcal F(t,0)=A(t)$ to $\mathcal F(t,\pi)=-A(t)$, with $\mathcal F(t,\tfrac{\pi}{2})=0$.
\end{proof}
The preceding lemma shows that, at fixed variance, no displacement can produce a magnitude of $Z$ or $X$ larger than the zero-displacement value $A(t)$. Equivalently, the extremal value of $|Z|$ at fixed $V_{qq}$ is $A(V_{qq})$, attained only for $q_0\in\sqrt{\pi}\mathbb Z$, while the extremal value of $|X|$ at fixed $V_{pp}$ is $A(V_{pp})$, attained only for $p_0\in\sqrt{\pi}\mathbb Z$. Thus, once the displacement dependence has been resolved, the remaining task is to understand how the extremal value $A(t)$ varies with the quadrature variance $t$. In particular, we will need $A(t)$ to be monotonic and invertible in order to relate values of $|X|$ and $|Z|$ directly to the corresponding quadrature variances. These properties are established by the following lemma.
\begin{lemma}[Monotonicity of $A$]\label{Lem:app-monotonicity} The function $A(t):(0,\infty)\to(0,1)$ is a continuous, strictly decreasing bijection, with $\lim_{t\to 0^+}A(t)=1$ and $\lim_{t\to\infty}A(t)=0$. 
\end{lemma}
\begin{proof}
The endpoint limits follow directly from the alternating series in Eq.~\eqref{Eq:app-function-A}. Its term magnitudes decrease with $m$, and the remainder after the term $m=N$ satisfies
\begin{equation}
    \left| A(t)-\frac{4}{\pi}\sum_{m=0}^{N} \frac{(-1)^m}{2m+1} e^{-\frac{\pi}{2}(2m+1)^2t} \right| \le\frac{4}{\pi(2N+3)}.
\end{equation}
Because this bound is independent of $t$, the limit $t\to0^+$ may be taken term by term. The usual series for $\arctan(1)$ then gives $\lim_{t\to0^+}A(t) =\frac{4}{\pi}\sum_{m=0}^{\infty}\frac{(-1)^m}{2m+1} =1$. The alternating-series bound also gives $0<A(t)\le(4/\pi)e^{-\pi t/2}$, and hence $A(t)\to0$ as $t\to\infty$. For $t>0$, term-by-term differentiation is justified by the uniform convergence of the differentiated series for $t\ge\varepsilon>0$. It gives
\begin{equation}
A'(t) =-2\sum_{m=0}^{\infty}(-1)^m(2m+1) e^{-\frac{\pi}{2}(2m+1)^2t} =-\partial_t\vartheta_1(0\mid2it).
\end{equation}
Evaluating the derivative of the product in Eq.~\eqref{eq:theta_product} at zero gives
\begin{equation}
A'(t) =-2e^{-\pi t/2} \prod_{n=1}^{\infty}(1-e^{-4\pi nt})^3 <0.
\label{eq:A_prime}
\end{equation}
Thus, $A$ is continuous and strictly decreasing. Together with the endpoint limits, this proves that it maps $(0,\infty)$ bijectively onto $(0,1)$.
\end{proof}

The two lemmas separate the roles of displacement and variance. Displacement determines where a component lies within its allowed interval, while the variance determines the size of that interval. We now combine these intervals with the uncertainty relation to obtain the full $XZ$ projection.

\subsubsection*{The exact reachable set and its boundary}
Let $\mathcal G_{XZ}$ denote the pairs $(X,Z)$ attainable by Gaussian states with finite covariance matrices. Limit points will be distinguished explicitly below. By Lemma~\ref{Lem:app-monotonicity}, the inverse function $\tau:=A^{-1}:(0,1)\to(0,\infty)$ is well defined and strictly decreasing. In particular, for $0<u<1$, the condition $u\le A(t)$ is equivalent to $t\le\tau(u)$. This is the relation we need to express the uncertainty constraint in terms of the observable components.
\begin{proposition} [Exact $XZ$ projection] \label{App:thm-exact-XZ} The exact reachable set is
\begin{equation}
\mathcal G_{XZ} =\bigcup_{\substack{a,b>0\\ab\ge1/4}} [-A(a),A(a)]\times[-A(b),A(b)].
\label{eq:prop_image}
\end{equation}
This set is convex. Its boundary with $X>0$ and $Z>0$ is the strictly concave curve
\begin{equation}
(X(t),Z(t)) =\left(A(t),A\left(\frac{1}{4t}\right)\right), 
\label{eq:exact_XZ_boundary}
\end{equation}
where $t>0$. A Gaussian state generates this boundary point exactly when its covariance, in $(q,p)$ order, is $V=\operatorname{diag}\left(\frac{1}{4t},t\right)$ and its displacements satisfy $q_0,p_0\in2\sqrt\pi\mathbb Z$. In particular, every such state is pure and has $V_{qp}=0$.
\end{proposition}
\begin{proof}
We first determine the allowed variances and the corresponding rectangles. We then find their outer boundary, establish convexity, and identify the states that attain the boundary.
\begin{enumerate}
    \item \emph{Allowed variances and reachable rectangles}. Write $a=V_{pp}$, $b=V_{qq}$, and $c=V_{qp}$. The covariance must satisfy $a,b>0$ and $ab-c^2=\det V\ge\tfrac14$. Consequently, $ab\ge\tfrac14$. Conversely, any positive $a,b$ satisfying this inequality are realized by the physical diagonal covariance $V=\operatorname{diag}(b,a)$. For these fixed variances, Lemma~\ref{App:lemm-fixed-variance} shows that independently varying $p_0$ and $q_0$ fills the rectangle $[-A(a),A(a)]\times[-A(b),A(b)]$. Taking the union over all allowed variances gives Eq.~\eqref{eq:prop_image}.

    \item \emph{Locating the boundary.} Suppose first that both components are nonzero. Any realization satisfies $|X|\le A(a)$ and $|Z|\le A(b)$, so $a\le\tau(|X|)$ and $b\le\tau(|Z|)$. Combining these inequalities with $ab\ge\tfrac14$ gives a necessary condition. It is also sufficient, because one may choose $a=\tau(|X|)$ and $b=\tau(|Z|)$ and use displacements to select the signs. Thus, for $0<|X|,|Z|<1$,
    \begin{equation}\label{eq:inverse_reachability}
        (X,Z)\in\mathcal G_{XZ} \Longleftrightarrow \tau(|X|)\tau(|Z|)\ge\frac14.
    \end{equation}
    Points on the axes are handled directly by Eq.~\eqref{eq:prop_image}, without evaluating $\tau(0)$. For example, every $(0,Z)$ with $|Z|<1$ belongs to a rectangle obtained by choosing $b>0$ with $A(b)\ge|Z|$ and then setting $a=1/(4b)$. The same reasoning applies to $(X,0)$, and the origin belongs to every rectangle. For a fixed $0<X<1$, Eq.~\eqref{eq:inverse_reachability} gives the largest possible positive value of $Z$ as $B(X):=A\left(\frac{1}{4\tau(X)}\right)$. Setting $t=\tau(X)$ gives Eq.~\eqref{eq:exact_XZ_boundary}. The endpoint limits of $A$ give the continuous extensions $B(0)=1$ and $B(1)=0$. Consequently, the entire finite-covariance set can be written as
    \begin{equation}\label{eq:graph_reachability}
        \mathcal G_{XZ} =\bigl\{(X,Z)\in(-1,1)^2: |Z|\le B(|X|)\bigr\}.
    \end{equation}
    Notice that no assumption of convexity has been used in obtaining this boundary.
    \item \emph{Establishing convexity.} We determine the slope of the boundary by differentiating Eq.~\eqref{eq:exact_XZ_boundary}, which gives $\frac{\d Z}{\d X} =-\frac{A'(1/(4t))}{4t^2A'(t)}$. To simplify this ratio, we use one identity relating the product in Eq.~\eqref{eq:A_prime} at reciprocal arguments. Define the Dedekind eta function on the positive imaginary axis by $\eta(iy):=e^{-\pi y/12} \prod_{n=1}^{\infty}(1-e^{-2\pi ny})$, where $y>0$. Then Eq.~\eqref{eq:A_prime} reads $A'(t)=-2\eta(2it)^3$. The identity $\eta(i/y)=\sqrt{y}\eta(iy)$ for $y>0$~\cite[Secs.~23.17--23.18]{DLMF}, applied with $y=2t$, gives $A'\left(\frac{1}{4t}\right) =(2t)^{3/2}A'(t)$. The slope therefore simplifies to $\frac{\d Z}{\d X}=-\frac{1}{\sqrt{2t}}$. Since $X=A(t)$ decreases with $t$, this slope becomes more negative as $X$ increases. Thus, $B$ is decreasing and strictly concave in $(0,1)$, and its continuous extension is concave on $[0,1]$. To establish convexity of the whole set, take two reachable points $(X_1,Z_1)$ and $(X_2,Z_2)$ and $0\le\lambda\le1$. The triangle inequality, concavity of $B$, and the fact that $B$ is decreasing give
    \begin{equation}
\begin{aligned}
|\lambda Z_1+(1-\lambda)Z_2|
&\le\lambda B(|X_1|)+(1-\lambda)B(|X_2|)\\
&\le B\bigl(\lambda|X_1|+(1-\lambda)|X_2|\bigr)\\
&\le B\bigl(|\lambda X_1+(1-\lambda)X_2|\bigr).
\end{aligned}
\end{equation}
Equation~\eqref{eq:graph_reachability} therefore places every point of the line segment between the two original points inside $\mathcal G_{XZ}$. This proves convexity.

\item \emph{Identifying the boundary states.} Suppose a Gaussian state generates the positive boundary point associated with $t$. The fixed-variance bounds give $A(t)\le A(a)$ and $A\left(\frac{1}{4t}\right)\le A(b)$. Since $A$ is strictly decreasing, $a\le t$ and $b\le1/(4t)$, which imply $ab\le1/4$. Physicality gives the reverse inequality and hence forces $a=t$, $b=1/(4t)$, and $c=0$, with $\det V=\tfrac14$. The covariance is therefore pure and diagonal, with the entries specified in the theorem. Both fixed-variance bounds must be saturated with positive signs. Lemma~\ref{App:lemm-fixed-variance} then requires $p_0,q_0\in2\sqrt\pi\mathbb Z$. Conversely, these covariance and displacement conditions generate the stated point.
\end{enumerate}
\end{proof}
The boundary in the other quadrants follows by changing the signs of $X$ and $Z$. The four axis endpoints $(\pm1,0)$ and $(0,\pm1)$ are limiting points, not outputs of finite-covariance states. Adding these four points gives the closure of $\mathcal G_{XZ}$. The boundary also has a direct squeezing interpretation. Writing $t=\tfrac12e^{2r}$ gives $X(r)=A\left(\tfrac12e^{2r}\right)$ and $Z(r)=A\left(\tfrac12e^{-2r}\right)$, with $\d Z/\d X=-e^{-r}$. Here $r$ is a signed parameter, its sign selects the squeezed quadrature, and $|r|$ is the squeezing strength. We now give two consequences of the theorem and the preceding lemmas. The first concerns the largest attainable value of $|X|+|Z|$.

\begin{corollary}[Maximum of $|X|+|Z|$]\label{cor:vacuum_mt} For Gaussian states with finite covariance,
\begin{equation}\label{eq:vacuum_XZ_maximum}
    \max_{(X,Z)\in\mathcal G_{XZ}}(|X|+|Z|) =2A\left(\tfrac12\right) =1.160314186\ldots.
\end{equation} 
Equality holds exactly when $V=\tfrac12 I_2$ and $q_0,p_0\in\sqrt\pi\mathbb Z$. In particular, the undisplaced vacuum is the unique maximizer among undisplaced Gaussian states.
\end{corollary}
\begin{proof}
By reflection symmetry, it is enough to maximize $X+Z$ in the positive quadrant. On the boundary, the derivative of $X+B(X)$ is $1+B'(X)$. The strict concavity of $B$ therefore places its maximum at the unique point with $B'(X)=-1$. The slope $\frac{dZ}{dX}=-\frac{1}{\sqrt{2t}}$ gives $t=\tfrac12$ and hence the value in Eq.~\eqref{eq:vacuum_XZ_maximum}. The maximizing value $t=\tfrac12$ corresponds to $r=0$ in the squeezing parametrization introduced above. The boundary-state characterization in Proposition~\ref{App:thm-exact-XZ} forces $V=\tfrac12 I_2$. Allowing all four sign choices for $X$ and $Z$ gives $q_0,p_0\in\sqrt\pi\mathbb Z$, as stated.
\end{proof}
Finally, the fixed-variance bounds give a direct criterion for certifying a minimum squeezing strength.
\begin{corollary}[Bound for limited squeezing] \label{cor:bounded_squeezing}
For a Gaussian state with squeezing strength $0\le r\le r_0$,
\begin{equation}
\max(|X|,|Z|) \le A\left(\tfrac12e^{-2r_0}\right).
\label{eq:witness_mt}
\end{equation}
The bound is attained by an undisplaced squeezed vacuum whose smaller quadrature variance is $\tfrac12e^{-2r_0}$. Consequently, within the Gaussian-state model, a value exceeding this bound certifies squeezing greater than $r_0$.
\end{corollary}
\begin{proof}
Write the covariance eigenvalues as $\nu e^{2r}$ and $\nu e^{-2r}$, where $\nu=\sqrt{\det V}\ge1/2$. Each quadrature variance is at least the smaller eigenvalue. Hence, if $r\le r_0$, $V_{qq}\ge\tfrac12e^{-2r_0}$ and $V_{pp}\ge\tfrac12e^{-2r_0}$. Combining Lemma~\ref{App:lemm-fixed-variance} with the strict
decrease of $A$ gives
\begin{equation}
\max(|X|,|Z|) \le\max\bigl(A(V_{pp}),A(V_{qq})\bigr) \le A\left(\tfrac12e^{-2r_0}\right).
\end{equation}
An undisplaced squeezed vacuum with $V_{qq}=\tfrac12e^{-2r_0}$ and $V_{pp}=\tfrac12e^{2r_0}$ attains equality through its $Z$ component.
\end{proof}

\subsection{Convergence of  $G_r$ to the stabilizer polytope in the $r\to\infty$ limit}\label{sec:proof_thm1}

\begin{proposition} 
[Convergence of the Gaussian image set to the stabilizer polytope]
\label{thm:gaussian_to_stab}
Let $G_r$ denote the set of single-qubit Bloch vectors $\vec{R}=(X,Y,Z)$ obtained from Gaussian states with squeezing parameter $r$, where $X$, $Y$, and $Z$ are given by Eqs.~\eqref{eq:Z_gaussian}, \eqref{eq:X_gaussian}, and \eqref{eq:Y_gaussian}. In the infinite-squeezing limit $r\rightarrow\infty$, the convex hull of the resulting set coincides with the single-qubit stabilizer polytope:   
\begin{equation}
    \operatorname{conv}(G_{\infty})  = \mathcal{P}_{\mathrm{STAB}}  = \operatorname{conv}\left\{(\pm1,0,0), (0,\pm1,0), (0,0,\pm1)\right\}.
\end{equation}
In particular, every limiting Gaussian Bloch vector lies on one of the three coordinate axes and satisfies 
\begin{equation}
    |X|\leq 1,\qquad |Y|\leq 1,\qquad |Z|\leq 1,
\end{equation}
while all six stabilizer states are attained in the limit.
\end{proposition}
\begin{proof}
   
We prove that at infinite squeezing the convex hull of the Gaussian states is the stabilizer polytope. For this we show that at infinite squeezing the Bloch vector $\vec{R}$ is either zero, or some point on an axis $x,y$ or $z$. Furthermore, the maximum/minimum value any component can take is $1/-1$.

Consider Eq~\eqref{eq:V_parameterized} for $r\to\infty$. The decaying exponentials vanish, and the covariance matrix entries read
\begin{align}
    &V_{qq} = \frac{e^{2r}}{4}(1-\cos(2\theta)),\nonumber\\
    &V_{pp} = \frac{e^{2r}}{4}(1+\cos(2\theta)),\label{eq:V_at_infinite_sq}\\
    &V_{qp} = -\frac{e^{2r}}{2}\sin(2\theta).\nonumber
\end{align}
We examine the exponential factors in Eq~\eqref{eq:Z_gaussian}, Eq~\eqref{eq:X_gaussian} and Eq~\eqref{eq:Y_gaussian},
\begin{align}
    &\exp\{-\frac{\pi}{2}(2m+1)^2V_{qq}\},\label{eq:Z_gaussian_exp}\\
    &\exp\{-\frac{\pi}{2}(2m+1)^2V_{pp}\},\label{eq:X_gaussian_exp}\\\
    &\exp\{-\frac{\pi}{2}[(2m+1)^2V_{qq}+(2n+1)^2V_{pp} +2V_{pq}(2m+1)(2n+1)]\}\label{eq:Y_gaussian_exp}.
\end{align}

\emph{The $Z$ component}. Note that $V_{qq}\to\infty$ when $r\to\infty$ for every angle such that $1-\cos(2\theta)\neq 0$. For the specific case $\theta = 0,\pi$, we have $V_{qq}=0$. This implies that the exponential~\eqref{eq:Z_gaussian_exp} $\exp\{-\frac{\pi}{2}(2m+1)^2V_{qq}\}=1$  for $\theta = 0,\pi$, while for any other angle $\exp\{-\frac{\pi}{2}(2m+1)^2V_{qq}\}\to0$.
Therefore we have 
\begin{equation}
    Z(q_0, \theta)\to\frac4\pi\sum_{m=0}^{\infty}\frac{(-1)^m}{2m+1}\cos((2m+1)\sqrt{\pi}q_0)(\mathbb{I}_{\{\theta = 0\}}+\mathbb{I}_{\{\theta = \pi\}}),
\end{equation}
where $\mathbb{I}$ denotes the indicator function.

\emph{The $X$ component}. With an analogous reasoning, we have that 
\begin{equation}
    X(p_0, \theta)\to\frac{4}{\pi}\sum_{m=0}^{\infty}\frac{(-1)^m}{2m+1}\cos((2m+1)\sqrt{\pi}p_0)\mathbb{I}_{\{\theta = \tfrac{\pi}{2}\}},
\end{equation}
where we have used now that $1+\cos(2\theta)=0$ for $\theta =\tfrac{\pi}{2}$.

\emph{The $Y$ component.}
In this case, we need to evaluate the argument $A =(2m+1)^2V_{qq}+(2n+1)^2V_{pp} +2V_{pq}(2m+1)(2n+1)$. Denote $M=2m+1$ and $N=2n+1$. Substituting ~\eqref{eq:V_at_infinite_sq}  and using $2\sin^2(\theta) = 1-\cos(2\theta)$, $2\cos^2(\theta) = 1+\cos(2\theta)$ and $\sin(2\theta)=2\sin(\theta)\cos(\theta)$ one arrives to 
\begin{equation}
    A = \frac{e^{2r}}{2}(M\sin\theta-N\cos\theta)^2.
\end{equation}
For $r\to\infty$, this expression is $0$ for the values $\tan\theta = N/M$ and diverges otherwise. Therefore we have 
\begin{equation}
Y(q_0,p_0,\theta) = -\frac{8}{\pi^2}\sum_{M>0\\ \text{ odd}}\sum_{N\text{ odd}}\frac{\cos(\sqrt{\pi}[Mq_0 +Np_0])}{MN}\mathbb{I}_{\{\theta = \tan^{-1}(\frac{N}{M})\}}.
\end{equation}
Putting all together, we see that the vector is $\vec{0}$ and it is found in the center of the Bloch sphere or one and only one of its components is different from zero, showing that it must be on an axis for $r\to\infty$. This is true since there exist no $\theta$ giving a true value for two or more of the indicators functions in the expressions above. 

Furthermore, by maximizing and minimizing $X$, $Z$ and $Y$, over $q_0$ and $p_0$ (taking the values $0$ and $\sqrt{\pi}$ for example), we find that 
\begin{equation}
    Z(\theta)\to\pm(\mathbb{I}_{\{\theta = 0\}}+\mathbb{I}_{\{\theta = \pi\}}), \quad X(\theta)\to\pm\mathbb{I}_{\{\theta = \pi/2\}},
\end{equation}
where we have used that $\sum_{m=0}^{\infty}\frac{(-1)^m}{2m+1} =\pi/4$. For $Y$ we have
\begin{equation}
Y(\theta) = \pm\frac{8}{\pi^2}\sum_{M>0\\ \text{ odd}}\sum_{N\text{ odd}}\frac{1}{MN}\mathbb{I}_{\{\theta = \tan^{-1}(\frac{N}{M})\}}.
\end{equation}
To maximize this expression, $Y(\theta)$ is non-vanishing if and only if $\tan\theta$ is the ratio of two odd integers. We can express this ratio in its lowest terms as
\begin{equation}
\tan\theta = \frac{n}{m},
\end{equation}
where $n>0$ and $m\neq 0$ are coprime odd integers ($\gcd(m,n)=1$). Consequently, all pairs $(N,M)$ that satisfy $\mathbb{I}_{\{\theta = \tan^{-1}(\frac{N}{M})\}} =1$ are simply positive odd multiples of this irreducible fraction, uniquely written as
\begin{equation}
M = km, \qquad N = kn,
\end{equation}
where $k$ is a positive odd integer. Substituting this parameterization into the equation collapses the double sum into a single sum over $k$:
\begin{align}
Y(\theta) = \pm\frac{8}{\pi^2} \frac{1}{mn} \sum_{\substack{k>0\\ k\ {\rm odd}}} \frac{1}{k^2}.
\end{align}
Using the known identity  $\sum_{\substack{k>0\\ k\ {\rm odd}}} \frac{1}{k^2} = \frac{\pi^2}{8},$ the prefactors cancel, yielding the expression:
\begin{equation}
Y(\theta) = \pm\frac{1}{mn}.
\end{equation}

To maximize $|Y(\theta)|$, we must minimize the integer denominator $|mn|$. Because $m$ and $n$ are non-zero odd integers, the minimum is $|m| = |n| = 1$. Therefore, the maximum is $\max_{\theta} |Y(\theta)| = 1$ which is attained along 
\begin{equation}
\tan\theta = \pm 1 \implies \theta = \frac{\pi}{4},\frac{3\pi}{4}.
\end{equation}
since $\theta\in[0,\pi]$. 

Putting everything together,  the extremal set of the Bloch vectors for high squeezing is exactly the six discrete stabilizer vertices. Therefore, its convex hull is $\operatorname{conv}(G_\infty) = \mathcal{P}_{\mathrm{STAB}}$ where we have defined $G_{\infty}=\lim_{r\to\infty}G_r$. As a consequence any value $W^G$ for a prepare-and-measure witness $W$ with stabilizer value $W^{\mathrm{STAB}}$ satisfies $W^G =W^{\mathrm{STAB}}$.
\end{proof}
\subsection{Rate of convergence of  $\operatorname{conv}(G_r)$ to $\mathcal P_{\rm STAB}$ for $r\to \infty$}\label{sm:thm_rate}
In this appendix we  show that $\operatorname{conv}(G_r)$ approaches $\mathcal P_{\rm STAB}$ at a rate that is not  uniform across the six stabilizer vertices: the coordinate $X$ and $Z$ vertices are approached doubly-exponentially fast in $r$, while the $Y$ vertices are approached only singly-exponentially fast, and it is this slower rate that controls the overall convergence.

Throughout we fix $q_0=p_0=0$, the extremal displacement established in in the main text, and write $Z(\theta,r)$, $X(\theta,r)$, $Y(\theta,r)$ for the Bloch components as functions of the squeezing angle $\theta$ at fixed squeezing $r$, using the untruncated series for $X$ and $Z$ in Eq.~\eqref{Eq:XYZ-component}.

\begin{lemma}[Double-exponential convergence in $X$ and $Z$]\label{lem:double_exp}
At $\theta=0$,
\begin{equation}
1-Z(0,r) = \frac4\pi e^{-r}\exp\left(-\frac\pi4 e^{2r}\right)\big(1+o(1)\big), \qquad r\to\infty, \label{eq:lemA_result}
\end{equation}
and, by the $q\leftrightarrow p$, $\theta\to\pi/2-\theta$ symmetry of Eqs.~\eqref{Eq:XYZ-component}, the identical rate holds for $1-X(\pi/2,r)$.
\end{lemma}
\begin{proof}
At $\theta=0$, $V_{qq}=\tfrac12(\cosh2r-\sinh2r)=\tfrac12e^{-2r}$, so
\begin{equation}
Z(0,r)=\frac4\pi F(t), \qquad \text{ with } F(t):=\sum_{m\ge0}\frac{(-1)^m}{2m+1}e^{-(2m+1)^2t}, \text{ and } t:=\frac\pi4e^{-2r}. \label{eq:F_def}
\end{equation}
Using the Fourier transform of a Gaussian $e^{-k^2t}=\frac1{2\sqrt{\pi t}}\int_{-\infty}^\infty e^{-x^2/4t}e^{ikx}dx$ and summing over $m$ before integrating (using dominant convergence of the series) we get,
\begin{equation}
F(t)=\frac1{2\sqrt{\pi t}}\int_{-\infty}^\infty e^{-x^2/4t}C(x)dx, \qquad \text{where }C(x)=\sum_{m\ge0}\frac{(-1)^m}{2m+1}\cos\big((2m+1)x\big), \label{eq:heat_kernel}
\end{equation}
since the odd sine part of the sum integrates to zero with the even Gaussian weight. 

Now, from the Maclaurin series $\arctan z = \sum_{m\ge0}\frac{(-1)^m}{2m+1}z^{2m+1}$, evaluated at $z=e^{ix}$, and using the identity $\tan(2u) = 2\text{Re}(z)/(1-|z|^2)$, where $z=e^{ix}$ and $u=Re(\arctan(z))$, we obtain
\begin{equation}
C(x)=\text{Re}(\arctan(e^{ix})) =\frac\pi4\operatorname{sgn}(\cos(x)), \label{eq:squarewave}
\end{equation}
resulting, in a $\pm\pi/4$ square wave of period $2\pi$ ($C(x)=\frac\pi4 \ \text{for}\ |x|<\frac\pi2,\quad C(x+\pi)=-C(x)$). Notice now that given that $\frac{1}{2\sqrt{\pi t}}\int_{-\infty}^{\infty}dx\exp(-x^2/4t) = 1$, then we can write
\begin{equation}
    \frac{4}{\pi}F(t) =1-\frac{4}{2\sqrt{\pi t}}\int_{x>0,\\cos(x)<0}e^{-x^2/4t}dx,
\end{equation}
where the coefficient $4$ accounts for the correction of assuming $\operatorname{sgn}(\cos(x))=1$ for all $x$ and for integrating only over the positive values of $x$.

To evaluate this limit as $t \to 0$ ($r \to \infty$), we extend the first integration interval to infinity to apply  asymptotic expansions for the Gaussian tail. Therefore, we rewrite the integral over the first negative region as:
\begin{equation}
    \int_{\pi/2}^{3\pi/2} e^{-x^2/4t} dx = \int_{\pi/2}^{\infty} e^{-x^2/4t} dx - \int_{3\pi/2}^{\infty} e^{-x^2/4t} dx.
\end{equation}
The remaining regions where $\cos(x) < 0$ occur at intervals $(\frac{5\pi}{2}, \frac{7\pi}{2}), (\frac{9\pi}{2}, \frac{11\pi}{2})$, and so forth. We can bound the sum of the subtracted tail at $3\pi/2$ and all subsequent positive contributions. Since the integrand is positive, the absolute value of this remainder is less than the integral over $[3\pi/2,\infty)$:
\begin{equation}
    \left| -\int_{3\pi/2}^{\infty} e^{-x^2/4t} dx + \sum_{k=1}^{\infty} \int_{\frac{(4k+1)\pi}{2}}^{\frac{(4k+3)\pi}{2}} e^{-x^2/4t} dx \right| < \int_{3\pi/2}^{\infty} e^{-x^2/4t} dx.
\end{equation}
Applying the standard upper bound for the Gaussian tail $\int_z^\infty e^{-x^2/4t} dx < \frac{2t}{z} e^{-z^2/4t}$ at $z = 3\pi/2$, shows that the remainder is bounded by $\frac{4t}{3\pi} e^{-9\pi^2/16t}$. Multiplying by the prefactor $\frac{2}{\sqrt{\pi t}}$, this total remainder yields an error of order $O\left(\sqrt{t}  e^{-9\pi^2/16t}\right)$. Therefore, 
\begin{equation}
    \frac{4}{\pi}F(t) = 1 - \frac{2}{\sqrt{\pi t}} \int_{\pi/2}^{\infty} e^{-x^2/4t} dx + O\left(\sqrt{t}  e^{-9\pi^2/16t}\right).
\end{equation}
Applying once more the standard asymptotic expansion for the Gaussian tail, $\int_z^\infty e^{-x^2/4t} dx = \frac{2t}{z} e^{-z^2/4t} \big(1 + o(1)\big)$, gives
\begin{equation}
    \frac{2}{\sqrt{\pi t}} \int_{\pi/2}^{\infty} e^{-x^2/4t} dx = \frac{2}{\sqrt{\pi t}} \left( \frac{4t}{\pi} e^{-\pi^2/16t} \right) \big(1+o(1)\big) = \frac{8}{\pi^{3/2}} \sqrt{t}  e^{-\pi^2/16t} \big(1+o(1)\big).
\end{equation}
Substituting this back into $F(t)$, the error term $O\big(e^{-9\pi^2/16t}\big)$ is much smaller than this leading term as $t \to 0$, allowing us to absorb it in the fast decaying terms $o(1)$ and get:
\begin{equation}
    1 - \frac{4}{\pi}F(t) = \frac{8}{\pi^{3/2}} \sqrt{t}  e^{-\pi^2/16t} \big(1 + o(1)\big).
\end{equation}
Finally, we substitute the squeezing parameter $t = \frac{\pi}{4} e^{-2r}$, ($\sqrt{t} = \frac{\sqrt{\pi}}{2} e^{-r}$ and $\pi^2/(16t) = \frac{\pi}{4} e^{2r}$), obtaining
\begin{equation}
    1 - Z = \frac{4}{\pi} e^{-r} \exp\left(-\frac{\pi}{4} e^{2r}\right) \big(1+o(1)\big),
\end{equation}
for $r \to \infty$. We obtain the same expression for $X$ by the symmetry $q_0\to p_0$ and $\theta\to\pi-\theta$.
\end{proof}
Now we move to the $Y$ direction.
\begin{lemma}[Single-exponential convergence in $Y$]\label{lem:single_exp}
At $\theta=\pi/4$,
\begin{equation}
1+Y\left(\frac\pi4,r\right) = \frac{2\sqrt2}{\pi}e^{-r}\big(1+o(1)\big), \qquad r\to\infty. \label{eq:lemB_result}
\end{equation}
\end{lemma}
\begin{proof}
    At $\theta=\pi/4$: $V_{qq}=V_{pp}=\tfrac12\cosh2r$, $V_{qp}=-\tfrac12\sinh2r$. Defining $M=2m+1$, $N=2n+1$,  the exponent of $\varphi_{m,n}$ in the $Y$ component appearing in Eq.~\eqref{Eq:XYZ-component} together with $\cosh2r-\sinh2r=e^{-2r}$ reads
\begin{equation}
M^2V_{qq}+N^2V_{pp}+2MNV_{qp} = \frac12(M-N)^2\cosh2r + MNe^{-2r}. \label{eq:offdiag}
\end{equation}
Off-diagonal terms ($M\ne N$) are therefore suppressed by $\exp\big(-\tfrac\pi4(M-N)^2\cosh2r\big)=O\big(e^{-e^{2r}}\big)$, doubly-exponentially smaller in $r$ than any single-exponential term, and their sum over $M\ne N$ remains subdominant.

The diagonal terms $M=N=k$ ($k$ odd, $k\ge1$) give an exponent $-\tfrac\pi2k^2e^{-2r}$, so
\begin{equation}
-Y\left(\frac\pi4,r\right) = \frac8{\pi^2}S(t'), \qquad S(t'):=\sum_{k\ \mathrm{odd}\ge1}\frac1{k^2}e^{-k^2t'}, \qquad t':=\frac\pi2e^{-2r}, \label{eq:S_def}
\end{equation}
up to the doubly-exponentially small correction from Eq.~\eqref{eq:offdiag}. Using $e^{-k^2u}/k^2=\int_u^\infty e^{-k^2v}dv$ and the classical Poisson-summation transform $\sum_{k\ge1}e^{-k^2u}=\tfrac12\sqrt{\pi/u}-\tfrac12+O\big(u^{-1/2}e^{-\pi^2/u}\big)$ as $u\to0^+$ so,
\begin{align}
    \sum_{k\ \mathrm{odd}\ge1}e^{-k^2u} &= \sum_{k\ge1}e^{-k^2u} -\sum_{k\ \mathrm{even}\ge1}e^{-k^2u} \nonumber, \\&= \tfrac12\sqrt{\pi/u}-\tfrac12+O\big(u^{-1/2}e^{-\pi^2/u}\big) - \{\tfrac12\sqrt{\pi/4u}-\tfrac12+O\big((4u)^{-1/2}e^{-\pi^2/4u}\big)\},\\
    &=\frac{1}{4}\sqrt{\pi/u} + O\big(u^{-1/2}e^{-c/u}\big).\label{eq:poisson_approx}
\end{align}
To evaluate the resulting integral we split the integration domain:
\begin{align}
    S(t') &= \int_{t'}^\infty \left( \sum_{k\ \mathrm{odd}\ge 1} e^{-k^2 v} \right) dv = \int_{0}^\infty \left( \sum_{k\ \mathrm{odd}\ge 1} e^{-k^2 v} \right) dv - \int_{0}^{t'} \left( \sum_{k\ \mathrm{odd}\ge 1} e^{-k^2 v} \right) dv.
\end{align}
The first integral corresponds to the evaluation at $t'=0$, which is a known convergent series:
\begin{equation}
    S(0) = \int_{0}^\infty \left( \sum_{k\ \mathrm{odd}\ge 1} e^{-k^2 v} \right) dv = \sum_{k\ \mathrm{odd}\ge 1} \frac{1}{k^2} = \frac{\pi^2}{8}.
\end{equation}
For the second integral, since $t' \to 0$ in the infinite-squeezing limit, we substitute the leading-order Poisson summation approximation (Eq.~\eqref{eq:poisson_approx}), $\sum_{k\ \mathrm{odd}} e^{-k^2 v} \approx \frac{1}{4}\sqrt{\pi/v}$:
\begin{equation}
    \int_{0}^{t'} \frac{1}{4}\sqrt{\frac{\pi}{v}}  dv = \frac{\sqrt{\pi}}{4} \Big[ 2\sqrt{v} \Big]_0^{t'} = \frac{\sqrt{\pi}}{2} \sqrt{t'}.
\end{equation}
Subtracting this integration from the exact constant $S(0)$ yields the final asymptotic behavior:
\begin{equation}
    S(t') = \frac{\pi^2}{8} - \frac{\sqrt{\pi}}{2}\sqrt{t'} + O\left(\sqrt{t'}e^{-c/t'}\right).
\end{equation}
Substituting into Eq.~\eqref{eq:S_def} and using $\sqrt{t'}=\sqrt{\pi/2}e^{-r}$,
\begin{equation}
-Y\left(\frac\pi4,r\right) = 1 - \frac4{\pi^{3/2}}\sqrt{t'} + \cdots = 1-\frac{2\sqrt2}\pi e^{-r}+\cdots,
\end{equation}
which rearranges to Eq.~\eqref{eq:lemB_result}.
\end{proof}
Finally we show that by convexity, the whole set must converge at the same rate than the extremal points in $X,Y$ and $Z$.
\begin{lemma}[Distance from a polytope with near-attained vertices]\label{lem:convexity}
Let $P=\operatorname{conv}\{\vec{v}_1,\ldots,\vec{v}_k\}$ and let $S$ be any set containing points $\vec{u}_1,\ldots,\vec{u}_k$ with $\|\vec{u}_i-\vec{v}_i\|\le\delta$ for each $i$. Then
\begin{equation}
\sup_{\vec{w}\in P} d\big(\vec{w},\operatorname{conv}(S)\big) \le \delta,
\end{equation}
where $d(\vec w,S) = \inf_{\vec v\in S}\|\vec w-\vec v\|_2$.
\end{lemma}
\begin{proof}
Any $\vec{w}\in P$ is $\vec{w}=\sum_i\lambda_i\vec{v}_i$ with $\lambda_i\ge0$, $\sum_i\lambda_i=1$. Set $\vec{w}'=\sum_i\lambda_i\vec{u}_i\in\operatorname{conv}(S)$; then $\|\vec{w}-\vec{w}'\|=\big\|\sum_i\lambda_i(\vec{v}_i-\vec{u}_i)\big\|\le\sum_i\lambda_i\|\vec{v}_i-\vec{u}_i\|\le\delta$.
\end{proof}
Therefore, for large $r$, Lemmas~\ref{lem:double_exp} and \ref{lem:single_exp} show  for each of the six vertices of $\mathcal P_{\mathrm{STAB}}$, a point of $G_r$ within Euclidean distance $O(e^{-r})$ or $O\big(e^{-r}\exp(-\tfrac\pi4e^{2r})\big)$, the former always dominating. In particular,  $\delta(r)=\tfrac{2\sqrt2}\pi e^{-r}(1+o(1))$ from the $Y$-vertices bounds all six distances at the stated order. Applying Lemma~\ref{lem:convexity} with $\delta=\delta(r)$ gives  the convergence rate stated in Eq.~\eqref{eq:thm_rate}.

\subsection{Set-magic and squeezing threshold}\label{sec: set magic}

Figure ~\hyperref[F-certification]{\ref{F-certification}(a)} shows that we can violate the $S_3^{\mathrm{STAB}}(t)$ for some $t$ and can hence certify set-magic. Here we make this observation exact by showing that for every $0<t<1$, there exist three pure Gaussian inputs with the same finite squeezing magnitude $r$ whose SSD outputs satisfy $S_3(t)>S_3^{\mathrm{STAB}}(t)$ and hence certify set magic. We write the bound as $S_3^{\mathrm{STAB}}(t)=\max_{1\le j\le3}C_j(t),$ where $(C_1,C_2,C_3)=\bigl(4-2t,\;2\sqrt5\,t+2\sqrt2(1-t),\;6t\bigr)$. Following the definitions introduced in  Sec.~\ref{sec:exact_XZ},  define $\alpha(r)=A(e^{-2r}/2)$ and $\beta(r)=A(e^{2r}/2)$. Using the $\vartheta_1(u\mid\tau)$ with $\vartheta_1'$ denoting the first-argument derivative give us the two identities, for $x,y>0$: $  -\partial_x A(x,0)=\vartheta_1'(0\mid2ix),$ and $\vartheta_1'(0\mid i/y)=y^{3/2}\vartheta_1'(0\mid iy)$. These follow from termwise differentiation and Jacobi's product and imaginary-transformation identities~\cite{DLMF}. Consequently $\alpha'(r)=-e^{-r}\beta'(r)=e^{-2r}\vartheta_1'(0\mid ie^{-2r})>0$ which on integration gives:
\begin{equation}
    0<1-\alpha(r) =\int_r^\infty e^{-s}[-\beta'(s)]\,ds <e^{-r}\beta(r).
\end{equation}
by identifying $A(0)=1$ and $A(\infty)=0$ using appropriate limits. The pure covariances $\frac12\operatorname{diag}(e^{2r},e^{-2r})$ and $\frac12\operatorname{diag}(e^{-2r},e^{2r})$, displaced by
$(m\sqrt\pi,n\sqrt\pi)$ with $m,n\in\mathbb Z$, have common squeezing $r$ and SSD Bloch vectors $(\pm\alpha,0,\pm\beta)$ and $(\pm\beta,0,\pm\alpha)$, with independent signs. Set $v_1=(\alpha,0,\beta)$ and choose one row:
\begin{equation}
\begin{array}{c|ccc}
 j&v_2&v_3&k_j(t)\\ \hline
 1&(-\alpha,0,\beta)&(-\beta,0,-\alpha)&4t\\
 2&(\beta,0,-\alpha)&(-\alpha,0,\beta)&4t/\sqrt5\\
 3&(\alpha,0,-\beta)&(-\alpha,0,-\beta)&4(1-t)
\end{array}
\end{equation}
Projective measurement axes parallel to $v_1+v_2-v_3$ and $v_1-v_2$ achieve 
\begin{equation}
 W_j=2t\|v_1+v_2-v_3\|+2(1-t)\|v_1-v_2\| \ge C_j(t)\alpha+k_j(t)\beta.
\end{equation}
For row $1$ and $3$ the above inequality can be checked via straightforward algebra, for $2$ choose the following inequalities:
$\sqrt{(2\alpha+\beta)^2+\alpha^2}\ge\sqrt5\,\alpha+2\beta/\sqrt5$ and $\sqrt{\alpha^2+\beta^2}\ge\alpha$ to realize the above bound. Further, choose $j$ with $C_j(t)=S_3^{\mathrm{STAB}}(t)$ and corresponding $k_j(t)$. Any finite $r\ge\log[C_j(t)/k_j(t)]$ then satisfies
\begin{equation}
\begin{aligned}
W_j-C_j(t)
&\ge k_j(t)\beta-C_j(t)(1-\alpha)\\
&>[k_j(t)-C_j(t)e^{-r}]\beta\ge0.
\end{aligned}
\end{equation}
This violation certifies set-magic for any value of $t$ via the explicit squeezing dependence $r\ge\log[S_3^{\mathrm{STAB}}(t)/k(t)]$ where $k(t)=k_j(t)$ appropriately paired. A consistent choice is always $k(t)=4\min\{t/\sqrt{5},1-t\}$.
\subsection{Analytical approximation for the RoM of Gaussian states under SSD}\label{sec:ROM_gaussian_States}
In this section we offer an analytical expression for the RoM, via the first term approximation of the Bloch vector series. First, we compare the RoM value for the vacuum state given by this approximation, with the exact value obtained in Ref.~\cite{Calcluth_2024} numerically, and proven analytically in Sec.~\ref{sec:exact_XZ}. Then, by maximizing the analytical approximation, we obtain a bound on the maximum RoM possible in this set, which is $1.302$. Together with the result in Sec.~\ref{sec:exact_XZ}, this constitutes an analytical derivation of the numerical values obtained in~\cite{Calcluth_2024}.

First, from Eqs.~\eqref{eq:Z_gaussian_approx}, \eqref{eq:X_gaussian_approx} and \eqref{eq:Y_gaussian_approx}, and following Eq.~\eqref{eq:single_qubit_ROM} one easily gets the approximated RoM of any Gaussian state under the SSD channel 
\begin{equation}
   R=  |X| +|Y|+|Z|\approx R^{(1)}(\hat\rho_L). \label{eq:ROM_approx}
\end{equation}
In particular, for the vacuum state, we have that $p_0=q_0=r=0$ and therefore, from Eqs~\eqref{eq:V_parameterized}, $V_{qq}=V_{pp}=1/2$ while $V_{qp}=0$.
resulting in 
\begin{equation}
   R(\text{vac})=\frac{8}{\pi}\exp\left(-\frac{\pi}{4}\right)\approx 1.161,
\end{equation}
which fairly agrees with the numerical value of $1.160$ obtained numerically in~\cite{Calcluth_2024} and analytically in Sec.~\ref{sec:exact_XZ}. Now we maximize the RoM~\eqref{eq:ROM_approx} over all the displacements, squeezing and rotations. First, note that to maximize it, we must take $p_0$ and $q_0$ to be multiples of $\sqrt{\pi}$. By symmetry, it is enough to consider $p_0=q_0=0$. This reduces $R$ to the symmetric  expression
\begin{align}
    R =& \frac{8}{\pi} \exp\left\{-\frac{\pi}{4}\cosh(2r)\right\} \cosh\left(\frac{\pi}{4}\cos(2\theta)\sinh(2r)\right)\nonumber\\
    +& \frac{16}{\pi^2} \exp\left\{-\frac{\pi}{2}\cosh(2r)\right\} \sinh\left(\frac{\pi}{2}\sin(2\theta)\sinh(2r)\right)\label{eq:Rpi_q0p0_zero}.
\end{align}
Note that the first term comes from $|X|+|Z|$ while the second term corresponds to $|Y|$. Now, we can take the partial derivative respect to $\theta$ to find that $\theta =\pi/4$ is a critical point. Taking it twice  shows that it is actually a maximum. By substituting this value of theta in the above equation and deriving respect to $r$ and equating to zero, we obtain the transcendental equation
\begin{align}
    \cosh(2r)\cosh\left(\frac{\pi}{2}\sinh(2r)\right) - \sinh(2r)\sinh\left(\frac{\pi}{2}\sinh(2r)\right) =\nonumber\\ \frac{\pi}{4}\sinh(2r)\exp\left\{\frac{\pi}{4}\cosh(2r)\right\}.
\end{align}
This equation can be solved using Mathematica or python and gives the value of $r\approx 0.262$. Substituting the values $(\theta=\pi/4,r=0.262)$ in $R$ gives the approximated maximum value of $R_{\max}\approx1.302$. Comparing again with the numerical value of $1.303$  found in Ref~\cite{Calcluth_2024}, reached by the same values of squeezing and rotation that those found in this work, we show that our Bloch expression predicts a maximum ROM within a relative error of  $O(10^{-4})$.

\subsection{Finite squeezing}\label{sec:finiteq_Squeezing}

In the main text we describe that for the value of $t\approx0.643$, the magic state that maximizes $S_3(t)$ corresponds to the vacuum state displaced by $\sqrt{\pi}$, $\hat Z_L\ket{0}$, the other two states correspond to stabilizer vertices, arising from  high squeezing ($r=1.5$) Gaussian states.  The Bloch vectors obtained numerically are
\begin{equation}
\begin{aligned}
\vec{v}_1 &= \left(0,0,1\right),\\
\vec{v}_2 &= \left(-0.5802,0,0.5802\right),\\
\vec{v}_3 &= \left(0,0,-1\right),
\end{aligned}
\end{equation}
where $\vec{v}_2$ corresponds to the image of $\hat Z_L \ket{0}$. Note that the state is a noisy version of the $H$-type state $\frac{1}{2}(-1,0,1)$~\cite{Heinrich2019}. The value for $S_3(t=0.643)$ for the case $r\leq r_0$, is  $S_3(t=0.643) \approx 3.91$ while for Gaussian states with $r=r_0$ the value is $S_3(t=0.643) \approx 3.88 =S_3^{\mathrm{STAB}}(t=0.643)$.

\subsection{Randomness certification from the $S_3(t)$ witness}
\label{sec:SM_randomness}

We next use the prepare-and-measure (PAM) scenario to certify randomness from the behavior obtained by mapping Gaussian states through the stabilizer subsystem decomposition (SSD) map. We consider three preparations, labelled by $x=0,1,2$, and two binary measurements, labelled by $y=0,1$. For a given Gaussian realization, the SSD map produces three logical qubit states with Bloch vectors $\vec r_x$. We choose the measurement directions as
\begin{equation}
    \vec q_0 = \frac{\vec r_0+\vec r_1-\vec r_2}{\|\vec r_0+\vec r_1-\vec r_2\|_2}, \qquad \vec q_1 = \frac{\vec r_0-\vec r_1}{\|\vec r_0-\vec r_1\|_2}. 
    \label{eq:SM_measurement_directions}
\end{equation}
The resulting correlators are therefore
\begin{equation}
    E_{xy} = \vec r_x\cdot\vec q_y, \qquad x=0,1,2,\quad y=0,1. 
    \label{eq:SM_correlators}
\end{equation}
With this choice of measurements, the tilted witness can be written directly in terms of the observed correlators as
\begin{equation}
    S_3(t) = 2t\left(E_{00}+E_{10}-E_{20}\right) + 2(1-t)\left(E_{01}-E_{11}\right). 
    \label{eq:SM_S3_correlators}
\end{equation}

For each value of the tilting parameter $t$ and squeezing bound $r_0$, we first maximize $S_3(t)$ over triples of single-mode Gaussian states satisfying
\begin{equation}
    0\leq r_x\leq r_0, \qquad x=0,1,2,
\end{equation}
where $r_x$ denotes the squeezing parameter of the corresponding Gaussian state. The optimization is performed over the displacements, squeezing parameters, and squeezing angles of the three preparations. The Bloch vectors entering Eq.~\eqref{eq:SM_correlators} are evaluated using the exact  expressions~\eqref{Eq:XYZ-component}.

The randomness is then quantified from the complete PAM behavior generated by each optimal Gaussian realization. In particular, for each optimal realization we fix all six correlators $\{E_{xy}\}_{x=0,1,2;y=0,1}$ to their Gaussian values and determine the largest guessing probability compatible with the PAM model by means of the semidefinite program introduced in~\cite{sarubi2026}. Thus, the SDP does not merely fix the value of the witness $S_3(t)$; it constrains the complete set of observed correlators.

We introduce two branches corresponding to Eve's two possible guesses, $e=0,1$. To bound the probabilities, we employ an adapted Navascués-Vértesi hierarchy for prepare-and-measure scenarios, first developed for the prepare-and-broadcast scenario~\cite{sarubi2026}. Specifically, the optimization is performed at the $L(1, 2)$ level of the moment relaxation, where preparation amplitudes are retained up to scalar degree $q=1$ and measurement operators are represented by words of maximum length $t=2$. Each branch is represented by a moment functional $L_e \in \mathcal{S}_{1,2}(w_e)$, whose associated positive semidefinite moment matrix is subject to the normalization ($w_0+w_1=1$), preparation constraints for a bounded dimension $d=2$, and the binary measurement algebra ($B_y^2 = \mathbb{I}$). Denoting the corresponding unnormalized joint probabilities by $P(b,e|x,y)$ and the branch weights by $w_e$, the guessing probability for the generation setting $(x_{\rm g},y_{\rm g})$ is obtained from
\begin{equation}
    P_{\rm guess} = \max \left[ P(b=0,e=0|x_{\rm g},y_{\rm g}) + P(b=1,e=1|x_{\rm g},y_{\rm g}) \right], 
    \label{eq:SM_pguess}
\end{equation}
subject to
\begin{equation}
    \sum_{e=0}^{1} E_{xy}^{(e)} = E_{xy}^{G}, \qquad x=0,1,2,\quad y=0,1, 
    \label{eq:SM_fixed_distribution}
\end{equation}
where $E_{xy}^{G}$ denotes the correlator obtained from the Gaussian realization. In our implementation, the six equalities in Eq.~\eqref{eq:SM_fixed_distribution} are imposed to numerical precision. The moment matrices are required to be positive semidefinite, while the preparation constraints and normalization conditions enforce consistency with the considered PAM model.

The conditional min-entropy of the output bit is then
\begin{equation}
    H_{\min}(B|E) = -\log_2 P_{\rm guess}. 
    \label{eq:SM_min_entropy}
\end{equation}
Throughout this work, we choose the generation setting
\begin{equation}
    (x_{\rm g},y_{\rm g})=(2,1). 
    \label{eq:SM_generation_setting}
\end{equation}

An important feature of this procedure is that the value of $S_3(t)$ does not, in general, uniquely determine the full observed behavior. Consequently, different Gaussian realizations attaining the same optimal witness value can give different sets of correlators and hence different values of the certified min-entropy. To account for this degeneracy, for each $(t,r_0)$ we evaluate the SDP for the Gaussian configurations found at the maximum of $S_3(t)$ and retain the smallest certified min-entropy, corresponding to the most conservative randomness estimate:
\begin{equation}
    H_{\min}^{\rm worst}(t;r_0) = \min_{\lambda\in\mathcal{M}_{t,r_0}} H_{\min} \left( \{E_{xy}^{G}(\lambda)\}_{x,y} \right), 
    \label{eq:SM_worst_case_entropy}
\end{equation}
where $\mathcal{M}_{t,r_0}$ denotes the set of Gaussian configurations attaining the maximum of $S_3(t)$ within the considered squeezing range.

If the optimized witness does not exceed the classical bound
\begin{equation}
    S_3^C(t)=\max\{6t,4-2t\},
\end{equation}
we assign zero certified randomness. Otherwise, the SDP provides a lower bound on the conditional min-entropy of the generated output bit. This procedure therefore combines the continuous-variable optimization over physically accessible Gaussian states with the discrete-variable PAM randomness certification based on the corresponding mapped behavior.

In Table~\ref{tab:optimal_states_randomness} we give the optimal states and measurements yielding the saturated randomness bounds at $r_0 = 1.5$ for $t = 0.5$.
\begin{table}[h!]
\centering
\begin{tabular}{c c c c c c}
\hline
\hline
Preparation & $q_0$ & $p_0$ & $r$ & $\theta$ & Bloch vector $\vec{r}_x = (X,Y,Z)$ \\
\hline
$x = 0$ & $1.7725$ & $-0.3607$ & $1.4287$ & $3.1416$ & $(0, 0, -1)$ \\
$x = 1$ & $0.9117$ & $-0.0746$ & $1.5000$ & $1.5717$ & $(1, 0, 0)$ \\
$x = 2$ & $0$ & $1.7725$ & $0.5478$ & $1.5708$ & $(-0.9396, 0, 0.1215)$ \\
\hline
\hline
\end{tabular}

\vspace{0.3cm}
\begin{tabular}{c c}
\hline
\hline
Measurement & Direction $\vec{q}_y$ \\
\hline
$y = 0$ & $(0.8657, 0.0, -0.5006)$ \\
$y = 1$ & $(-0.7071, 0.0, -0.7071)$ \\
\hline
\hline
\end{tabular}
\caption{Parameters of the Gaussian states, corresponding logical Bloch vectors, and effective measurement directions optimizing $S_3(t)$ and minimizing the conditional min-entropy for the generation setting $(x_{\rm g}, y_{\rm g}) = (2, 1)$ at $r_0 = 1.5$.}
\label{tab:optimal_states_randomness}
\end{table}
For this specific configuration, the tilted witness achieves a value of $S_3(0.5) \approx 3.6547$. It is instructive to compare the certified randomness against the bare marginal predictability of the state. At the generation setting, Bob's most likely outcome occurs with probability $P(b=0|x_{\rm g}=2, y_{\rm g}=1) = (1 + E_{21})/2 \approx 0.7892$, which would naively suggest an apparent entropy of roughly $0.341$ bits. However, accounting for Eve's optimal strategy and side-information via the SDP hierarchy yields a strictly higher guessing probability of $P_{\rm guess} \approx 0.8828$. This reduces the  certified conditional min-entropy to $H_{\min}(B|E) \approx 0.1798$ bits, highlighting the importance of the full adversarial optimization. 

Finally, as a method note, the worst-case entropy in Eq.~\eqref{eq:SM_worst_case_entropy} is evaluated numerically by retaining, at each $(t,r_0)$, only the top-$k$ near-degenerate local optima of $S_3(t)$ found via multi-start L-BFGS-B, rather than an exhaustive characterization of the full degenerate manifold $\mathcal{M}_{t,r_0}$. The code can be found in~\cite{zamora2026code}.
\end{document}